\documentclass{article}

\usepackage{arxiv}

\usepackage[utf8]{inputenc} 
\usepackage[T1]{fontenc}    
\usepackage{hyperref}       
\usepackage{url}            
\usepackage{booktabs}       
\usepackage{amsfonts}       
\usepackage{nicefrac}       
\usepackage{microtype}      
\usepackage{lipsum}		
\usepackage{graphicx}
\usepackage{natbib}
\usepackage{doi}

\usepackage{bm,enumerate,amsmath,amsthm,rotating,multirow,tabulary,setspace}
\makeatletter
\newcommand*{\centerfloat}{%
  \parindent \z@
  \leftskip \z@ \@plus 1fil \@minus \textwidth
  \rightskip\leftskip
  \parfillskip \z@skip}
\makeatother
\newcommand\ind{\protect\mathpalette{\protect\independenT}{\bot}}
\def\independenT#1#2{\mathrel{\rlap{$#1#2$}\mkern2mu{#1#2}}}

\let\abs=\envert
\newtheorem{theorem}{Theorem}[section]
\newtheorem{corollary}{Corollary}[theorem]
\newtheorem{lemma}[theorem]{Lemma}
\newtheorem{definition}{Definition}[section]

\title{Recoverability and estimation of causal effects under typical multivariable missingness mechanisms}

\author{%
Jiaxin Zhang$^{1,2,*}$ \quad S. Ghazaleh Dashti$^{1,2}$ \quad John B. Carlin$^{1,2}$ \quad Katherine J. Lee$^{1,2}$ \quad Margarita Moreno-Betancur$^{1,2}$ \\
$^1$ Clinical Epidemiology and
Biostatistics Unit, Department of Paediatrics, University of Melbourne \\ 
$^2$ Clinical Epidemiology and Biostatistics Unit, Murdoch Children’s Research Institute \\
\texttt{jiaxin.zhang@mcri.edu.au}
}

\renewcommand{\shorttitle}{Causal inference in the presence of multivariable missing data}

\hypersetup{
pdftitle={A template for the arxiv style},
pdfsubject={q-bio.NC, q-bio.QM},
pdfauthor={David S.~Hippocampus, Elias D.~Striatum},
pdfkeywords={First keyword, Second keyword, More},
}

\begin{document}
\maketitle

\begin{abstract}
In the context of missing data, the identifiability or ``recoverability’’ of the average causal effect (ACE) depends not only on the usual causal assumptions but also on missingness assumptions that can be depicted by adding variable-specific missingness indicators to causal diagrams, creating ``missingness directed acyclic graphs’’ (m-DAGs). Previous research described ten canonical m-DAGs, representing typical multivariable missingness mechanisms in epidemiological studies, and mathematically determined the recoverability of the ACE in each case. However, this work assumed no effect modification and did not investigate methods for estimation across these scenarios. Here we extend this research by determining the recoverability of the ACE in settings with effect modification and conducting a simulation study evaluating the performance of widely used missing data methods when estimating the ACE using correctly specified g-computation, which has not been previously studied. Methods assessed were complete case analysis (CCA) and various implementations of multiple imputation (MI) with varying degrees of compatibility with the outcome model used in g-computation. Simulations were based on an example from the Victorian Adolescent Health Cohort Study (VAHCS), where interest was in estimating the ACE of adolescent cannabis use on mental health in young adulthood. We found that the ACE is recoverable when no incomplete variable (exposure, outcome or confounder) causes its own missingness, and non-recoverable otherwise, in simplified versions of the canonical m-DAGs that excluded unmeasured common causes of missingness indicators. Despite this, simulations showed that MI approaches that were compatible with the outcome model in g-computation may enable approximately unbiased estimation across all canonical m-DAGs, except when the outcome causes its own missingness or it causes the missingness of a variable that causes its own missingness. In the latter settings, researchers need to consider sensitivity analysis methods incorporating external information (e.g. delta-adjustment methods). The VAHCS case study illustrates the practical implications of these findings.
\end{abstract}


\section{Introduction}
In the presence of missing data, the average causal effect (ACE) is identifiable or ``recoverable’’ if it can be consistently estimated from observable data using an appropriate procedure \citep{mohan2013graphical,mohan2014graphical,mohan2021graphical}, i.e., if there exists an ACE estimator that converges to the true ACE with increasing sample size. The recoverability of the ACE depends on the missingness mechanism. The traditional framework of Rubin \citep{rubin2004multiple} classifies missingness mechanisms as ``missing completely at random'' (MCAR), ``missing at random'' (MAR) and ``missing not at random'' (MNAR) categories. An assumption of MAR is widely considered to provide a useful way of identifying settings where missing data methods such as multiple imputation (MI) enable unbiased estimation. However, this approach to analysis planning is problematic for two reasons. Firstly, it is difficult to understand or assess the MAR assumption in settings with multivariable missingness, which is the norm in epidemiological studies \citep{seaman2013meant}. Also, the MAR assumption is a sufficient but not necessary condition for consistent estimation with methods such as MI \citep{rubin2004multiple,seaman2013meant,little2019statistical}.

Recently, Mohan and colleagues proposed missingness directed acyclic graphs (m-DAGs), where variable-specific missingness indicators are added to causal DAGs, as an alternative, more intuitive framework for specifying assumptions regarding the missingness mechanisms \citep{mohan2013graphical}. These graphs allow clear depiction of the assumed causes of missingness in each incomplete variable \citep{mohan2014testability,thoemmes2015graphical,tian2015missing}. Furthermore, Mohan and colleagues investigated graphical conditions for the recoverability of causal effects in general settings \citep{mohan2014graphical,shpitser2015missing,mohan2021graphical}. Moreno-Betancur and colleagues used these theoretical results and other derivations to determine the recoverability of the ACE in the absence of effect modification in ten canonical m-DAGs, representing typical multivariable missingness mechanisms in point-exposure epidemiologic studies \citep{moreno2018canonical}. However, effect modification is likely in such studies and extending these results to the more general setting is needed. 

Once the recoverability of the ACE has been determined, the challenge remains in identifying an appropriate procedure for handling missing data in the context of the chosen approach for causal effect estimation. One such approach is g-computation, which extends multivariable regression to allow for effect modification by including exposure-confounder interactions in an outcome regression model. The ACE is estimated by contrasting the estimated average of the predicted potential outcomes from this model when setting all observations to exposed and nonexposed \citep{robins1986new,snowden2011implementation,Causal2020}. Despite it being a fundamental and practically accessible method for causal effect estimation, there has been scant research into approaches for handling missing data with g-computation. Discussion on the relative performance of widely used approaches, such as complete-case analysis (CCA) (excluding records with missing values in any of the analysis variables) and MI, is primarily limited the problematic MCAR/MAR/MNAR framework and generalised linear regression \citep{sterne2009multiple,lee2021framework,zhang2022should}. However, beyond the need to consider other missingness assumption frameworks, it is increasingly recognised that unbiased estimation with MI relies on it being tailored so that it is ``compatible’’ with the analysis method. 

Compatibility of MI is usually described as the existence of an overarching joint distribution for both the imputation model and the ``analysis model’’, which is suited to the usual context where the target analysis consists in simply fitting a model \citep{meng1994multiple,bartlett2015multiple1}. In these settings, achieving compatibility may require incorporating in MI the interactions in the analysis model \citep{von20098,seaman2012multiple,goldstein2014fitting,tilling2016appropriate,zhang2022should}. For instance, in a ``fully conditional specification (FCS)’’ approach \citep{van2007multiple,white2011multiple}, the univariate imputation models would include these interaction terms. An alternative way to achieve compatibility is the ``substantive model compatible (SMC)’’ FCS approach \citep{bartlett2015multiple1,bartlett2015multiple2}, proposed by Bartlett and colleagues, where the imputation model is obtained by multiplying the substantive (analysis) model with the conditional distribution for each imputed variable. However, how MI should be tailored for g-computation, an estimation method incorporating a two-stage process of regression and prediction, and how important this is, has not been studied. 

This paper aims to (1) investigate the recoverability of the ACE in the presence of effect modification under typical missingness mechanisms depicted by m-DAGs, and (2) evaluate the performance of CCA and six possible implementations of MI tailored to outcome model in g-computation to estimate the ACE. The paper is organised as follows. Section 2 introduces a case study from the Victorian Adolescent Health Cohort Study (VAHCS) \citep{patton2002cannabis}, which was used as a basis for the simulation study. Section 3 provides theoretical results regarding the recoverability of the ACE in ten simplified canonical m-DAGs. Section 4 describes the simulation study investigating the performance of missingness approaches across a range of data generation models and missingness scenarios. Section 5 describes the results of the case study using the VAHCS data. Finally, section 6 summarises and discusses our findings.

\section{Case study}
The VAHCS is a population-based longitudinal cohort study that recruited 1943 participants aged 14-15 years old (1000 females) from Victorian schools in 1992-1993. The study was approved by the Human Research Ethics Committee of the Royal Children’s Hospital, Melbourne, Australia. Participants were surveyed every six months for three years following recruitment (waves 2 to 6, adolescence phase) and then again at wave 7 in 1998 (the young adulthood phase). The causal question of interest here, drawn from \citep{patton2002cannabis}, aimed to estimate the extent to which an intervention to prevent frequent cannabis use in female adolescents would change their risk of depression and anxiety in young adulthood. The exposure measure was self-reported frequent cannabis use in adolescence, where a participant was considered exposed if they reported use of cannabis more than once a week at any adolescent wave, and unexposed otherwise. The outcome (mental health score at wave 7 (age 20-21)) was the log-transformed and standardised Computerised Revised Clinical Interview schedule (CIS-R) \citep{lewis1992manual}. The confounders, selected on substantive grounds, were parental education, parental divorce or separation, antisocial behaviour, depression and anxiety, and frequent alcohol use, all measured across waves 2 to 6. After excluding 39 records with incomplete data for parental education, parental divorce and antisocial behaviour, the proportion exposed was 8.7\% among $n$=961 female participants, which henceforth defines the analytical sample. Table \ref{tab1} shows descriptive statistics and missing data proportions in this sample for each analysis variable as well as for participant's age at wave two (log-transformed and standardised), which will serve as an auxiliary variable for MI. The complete case proportion was 65.3\%.

\section{Recoverability of the ACE in canonical m-DAGs}
\subsection{Notation, definitions and assumptions}
We use $X$ and $Y$ to represent the exposure and outcome, both potentially subject to missing data. Confounders ($\bm{Z}$) are grouped into two sets, a vector of completely observed confounders ($\bm{Z_1}$) and a vector of confounders with missing data ($\bm{Z_2}$). The missingness indicators for exposure, outcome and incomplete confounders are denoted by $M_X$, $M_Y$ and $M_{\bm{Z_2}}$, respectively. Each is a binary indicator coded 1 if there are missing values in the corresponding variable and 0 otherwise. Finally, $\bm{U}$ represents a vector of unmeasured common causes of the exposure and confounders and $\bm{W}$ a vector of unmeasured common causes of the missingness indicators.

The ACE is defined as the difference in the expected value of the potential outcomes under exposure versus no exposure, which can be expressed as $ACE={\rm E}(Y\vert do(X=1))-{\rm E}(Y\vert do(X=0))$, where the $do(X=x)$ operator represents the intervention to set $X$ to $x$ \citep{pearl2009causality}.  

In the absence of missing data, under the identifiability assumptions of exchangeability given $\bm{Z}$, consistency, and positivity, the ACE is identified by \citep{Causal2020}: 
\begin{equation}\label{ACE}
    ACE=\sum_{\bm{z}} {\rm E}(Y\vert \bm{Z}=\bm{z},X=1){\rm P}(\bm{Z}=\bm{z})-\sum_{\bm{z}} {\rm E}(Y\vert \bm{Z}=\bm{z},X=0){\rm P}(\bm{Z}=\bm{z}).
\end{equation}

In the presence of missing data, the identifiability of the ACE is termed recoverability and depends on further assumptions regarding the missingness mechanism. Determining the recoverability of the ACE requires expressing the ACE in terms of the distribution of the observable variables, where every appearance of a partially observed variable in the expression occurs in a probability function conditional on the variable being observed, i.e., on its missingness indicator being equal to 0 \citep{mohan2013graphical}.

Figure \ref{m-DAG} illustrates ten canonical m-DAGs that depict typical missingness mechanisms in an epidemiological setting, characterised as ``canonical’’ by \citep{moreno2018canonical}. The assumptions encapsulated in the canonical m DAGs are summarised in the Appendix. Briefly, the m-DAGs differ in the presence of arrows from $X$ and $\bm{Z_2}$ to their own missingness indicators, and from $Y$ to its own missingness indicator. Note, for m-DAGs with an arrow from $\bm{Z_2}$ to $M_{\bm{Z_2}}$, we assume that there is an arrow from at least one incomplete confounder to its own missingness indicator. 

\subsection{Results for recoverability of the ACE in canonical m-DAGs}\label{the.res.section}
We determined the recoverability of the ACE in the m-DAGs of Figure \ref{m-DAG} in the presence of effect modification by investigating the recoverability of the marginal distribution of the potential outcome using $do$-calculus. For the theoretical work in this section only, we assumed the absence of unmeasured common causes for missingness indicators ($\bm{W}$), corresponding to the m-DAGs with the dashed arrows removed. Table \ref{Recoverability} presents the recoverability results for the ACE and the expression for the marginal distribution of the potential outcomes in terms of observable data for settings where it is recoverable, in the case of binary variables, with formal mathematical proofs provided in the Appendix. Briefly, the ACE is recoverable for m-DAGs where no variable causes its own missingness (m-DAGs A, B and C), and non-recoverable in m-DAG D where there are arrows from $X$ and $\bm{Z_2}$ to their own missingness indicators. As a corollary, the ACE is also non-recoverable in m-DAGs E, F, I and J, which are constructed by adding arrows to m-DAG D (Lemma 4 of \citep{mohan2013graphical}). We conjecture that the ACE is also non-recoverable in m-DAGs G and H (where there is an arrow from $Y$ to its own missingness indicator), although this has not been formally proven. Evidently, m-DAGs D-J are also non-recoverable if there are unmeasured common causes for missingness indicators ($\bm{W}$).

\section{Simulation study}
In this section, we describe a simulation study to assess the performance of CCA and various MI approaches when estimating the ACE using correctly specified g computation across a range of data generation and missingness scenarios, the latter depicted by the m-DAGs of Figure \ref{m-DAG} allowing for the presence of unmeasured common causes of the missingness indicators $\bm{W}$ (the m-DAGs including the dashed arrows). We generated data to mimic the marginal distributions and patterns of association seen in the VAHCS case study in Table \ref{tab1}. 

\subsection{Complete data generation}
Parameter values for the data generation models described next are provided in the Supplementary Material and were based on estimates from the VAHCS data unless stated otherwise. We simulated complete datasets based on Figure \ref{simuDAG} and generated the variables in sequence as follows. We first generated the auxiliary variable $A$ from the standard normal distribution and confounder $C_1$ from a Bernoulli distribution with the same prevalence as in the case study (Table \ref{tab1}). Next, we generated confounders $C_2-C_5$ from Bernoulli distributions in sequence, where the success probabilities were given by logistic regression models conditional on the previously generated variables. Then we generated the exposure $X$ from a Bernoulli distribution with success probability given by a logistic regression on the auxiliary variable and all confounders, considering both low (10\%) and high (50\%) prevalence scenarios by changing the intercept. Given that the confounders and exposure were generated sequentially conditional on previously generated confounders and the auxiliary variable, they were correlated with each other, which (for the purpose of recoverability) is equivalent to having an unmeasured common cause $U$ as in the m-DAGs of Figure \ref{m-DAG}. Lastly, we generated the outcome $Y$ using the following regression model: 
\begin{equation}
\begin{aligned}\label{outgen}
    {\rm E}(Y)&=\beta_0+\beta_1 C_1+\beta_2 C_2+\beta_3 C_3+\beta_4 C_4+\beta_5 C_5+\beta_6 X+\beta_7 C_3X+\beta_8 C_4X\\
    &+ \beta_{1,4} C_1C_4+\beta_{2,4} C_2C_4+ \beta_{3,4} C_3C_4+\beta_{4,5} C_4C_5+\beta_{3,5} C_3C_5,
\end{aligned}
\end{equation}

Six outcome scenarios were considered, where the strength of the exposure-confounder interaction terms relative to the main effect was null, low ($\pm$0.5 times in outcome scenarios II-IV) or extreme ($\pm$3 times in outcome scenarios V and VI) - see Table \ref{outscen}. The true value of the ACE, the target estimand (see Section \ref{target} below), was fixed to 0.3 across all scenarios by modifying the parameter $\beta_6$ in the model (\ref{outgen}). In the presence of exposure-confounder interactions, the required value of $\beta_6$ was found by iterating over a grid, where we generated a large dataset for each grid value and estimated the ACE via g-computation using a correctly specified outcome model. 

The sample size for each of the six outcome scenarios was determined to achieve 80\% power for the ACE in each case to make the results comparable across the scenarios. As a result, the sample sizes were 1400, 2200, 2000, 2700, 2200 and 2000 for outcome scenarios I to VI, respectively, in the setting with 10\% exposure prevalence and 700 for the six outcome scenarios in the 50\% exposure prevalence setting. 

\subsection{Missingness generation}
We assumed the auxiliary variable $A$ was complete and generated missing data in the exposure ($X$), outcome ($Y$), and two confounders ($C_4$ and $C_5$). We generated the missingness indicators following the m-DAGs in Figure \ref{m-DAG}, where $\bm{Z_2}$ represented the set of incomplete confounders, $C_4$ and $C_5$, and $\bm{Z_1}$ represented the set of complete confounders, $C_1$, $C_2$ and $C_3$. We generated a single variable $W$ from the standard normal distribution to incorporate unmeasured common causes of the missingness indicators. For each m-DAG, we designed five missingness scenarios depending on whether they included interactions between exposure and other variables with different properties in the generation of missingness indicators, as follows:
\begin{enumerate}[(i)]
\item No interaction terms
\item Exposure interaction with the strong complete confounder ($XC_3$ interaction)
\item Exposure interaction with the strong incomplete confounder ($XC_4$ interaction)
\item Exposure interaction with the weak incomplete confounder ($XC_5$ interaction)
\item Exposure interaction with the outcome ($XY$ interaction)
\end{enumerate}

Further details for missingness generation are provided in the Supplementary Material. The missingness proportion for each incomplete variable in each scenario was 15\% for $C_4$ and $C_5$, and 20\% for $X$ and $Y$. This resulted in approximately 55\% of observations having complete data.

\subsection{Target estimand and analysis}\label{target}
The target estimand in our study was the ACE, estimated using g-computation based on a correctly specified outcome model, i.e. identical to the outcome generation model \ref{outgen} for each outcome scenario. In the absence of missing data, the ACE can be consistently estimated by g-computation under the aforementioned identifiability assumptions if the outcome model is correctly specified. The standard error can be estimated using the bootstrap. In the CCA and MI approaches described below, we used 240 bootstrap samples. The number of bootstrap samples was based on Andrews' three-step method \citep{hall1986number,andrews2000three,andrews2001evaluation}. 

\subsection{Methods for handling missing data}
The methods we considered for handling missing data were broadly categorised as CCA, the MI-SMC approach and the MI-FCS approaches (all other MI approaches), as listed in Table \ref{methods}. For each incomplete variable, MI approaches used linear regression for the continuous outcome and logistic regression for the binary variables (the confounders and the exposure). 

For MI-FCS approaches, all variables in the outcome model for g-computation as well as the auxiliary variable $A$ were used as predictors in the relevant univariate imputation models. For the MI-EO, MI-EI and MI-Com approaches (three of the MI-FCS approaches), the imputation model was tailored to the target analysis by including the exposure-confounder interaction(s) that were present in the outcome generation model (same as that used in g-computation), i.e. $XC_3$ interaction was included in the imputation model for imputing incomplete confounders and outcome in outcome scenarios II and V. For the same reason, the $XC_4$ interaction was included in the imputation model in outcome scenarios III and VI, and both $XC_3$ and $XC_4$ interactions were included in the imputation model in outcome scenario IV. As recommended in the literature, interactions in the imputation model were updated from the imputed variables at each FCS cycle, referred to as the improved passive FCS approach \citep{seaman2012multiple}. For the MI-SMC approach, the substantive model for imputation was an expanded version of the outcome generation model including the auxiliary variable $A$. 

All MI-FCS approaches were carried out using the {\tt{mice}} command in {\tt{R}} \citep{buuren2010mice}. The MI-SMC approach was carried out using the {\tt{smcfcs}} package in {\tt{R}} \citep{bartlett2015multiple1,bartlett2015multiple2}. For each MI approach, 5 imputed datasets were generated, each obtained after the default 5 iterations. The target analysis was applied to each imputed dataset, and ACE estimates and its standard error from bootstrapping were pooled using Rubin’s rule \citep{rubin2004multiple}.

\subsection{Evaluation criteria}
We simulated 2000 datasets for each scenario. For each method, we report the bias, defined as the difference between the mean of the ACE estimates and the true value (0.3), in both absolute and relative terms (i.e. as a percentage); the empirical standard error (empirical SE), given by the square root of the variance of the 2000 estimates; the estimated standard error (estimated SE), given by the average of 2000 estimated standard errors; and the coverage probability, estimated by the proportion of the 95\% confidence intervals (CI) that contained the target value (0.3) across the 2000 datasets. Monte Carlo standard errors (MCSE) are also reported for all measures \citep{morris2019using}.

\subsection{Simulation results}\label{simu.res.section}
Figure \ref{RB} shows the relative bias of the missing data methods in estimating the ACE in the ten m-DAGs. We will summarise the performance in the 50\% exposure prevalence scenario first, then comment on the 10\% exposure prevalence. We categorise the m-DAGs into three groups based on the presence of arrows from incomplete variables to their own missingness indicator.  

\paragraph{m-DAGs where no variable causes its own missingness: A, B, and C.} The theoretical work showed the ACE was recoverable in these m-DAGs when there was no unmeasured common cause of missingness indicators (i.e. no $\bm{W}$). The simulation results showed that MI methods enabled approximately unbiased estimation in most of these settings, even in the presence of $W$. Specifically, for m-DAG A, CCA and all MI methods were approximately unbiased across all outcome scenarios ($\abs{RB}$ $<$ 6\%), and in m DAG B for outcome scenarios I-V ($\abs{RB}$ $<$ 8\%). In comparison, for m-DAG B in outcome scenario VI (strong interaction between exposure and incomplete confounder), all methods were somewhat biased ($\abs{RB}$ 7-34\%). For m-DAG C, all MI-based methods had small bias for outcome scenarios I-V ($\abs{RB}$ $<$ 15\%), and high bias for outcome scenario VI ($\abs{RB}$ 14-45\%). CCA had high bias for m-DAG C across all outcome scenarios ($\abs{RB}$ $>$ 62\%). Across all settings, the performance of MI-Sim, MI-EO, MI-EI and MI-EC were similar to each other. These methods were more biased than MI-Com and MI-SMC, with MI-SMC being the least biased in outcome scenario VI. 

\paragraph{m-DAGs where the exposure and confounder (but not the outcome) cause their own missingness: D, E, F, and I.} The theoretical work (assuming that there was no $W$) suggested the ACE was non-recoverable in these m-DAGs. However, the simulation results showed that approximately unbiased estimation may still be possible in some of these settings, even in the presence of $W$. For m-DAG D, all missing data methods were approximately unbiased across all outcome scenarios ($\abs{RB}$ $<$ 6\%), except for CCA in outcome scenario VI ($\abs{RB}$ 25-26\%). For m-DAG E, all missing data methods had small biases in outcome scenarios (i)-(v) ($\abs{RB}$ $<$ 14\%, all upward biases), and in comparison, larger biases in outcome scenario (vi) ($\abs{RB}$ 9-34\%, all downward biases). MI-SMC was the exception, returning estimates with small bias ($\abs{RB}$ $<$ 4\%) across all outcome scenarios. For m-DAG F, all missing data methods had high bias for all outcome scenarios ($\abs{RB}$ $>$ 22\%). Bias was highest in missingness scenario (v) ($\abs{RB}$ $>$ 66\%), where exposure-outcome interaction was used in generating missingness in the exposure. For m-DAG I, the performance of the missing data methods was like that for m-DAG F, although MI approaches returned estimates with somewhat smaller biases (overall $\abs{RB}$ $>$ 15\%; $\abs{RB}$ $>$ 32\% in outcome scenario VI; $\abs{RB}$ $>$ 50\% in missingness scenario (v); $\abs{RB}$ $>$ 67\% for CCA in all scenarios) than those from CCA. 

\paragraph{m-DAGs where the outcome causes missingness for itself: G, H and J.} The theoretical work showed that the ACE was non-recoverable in these m-DAGs with no $W$. In the simulations, all missing data methods were considerably biased ($\abs{RB}$ $>$ 32\%) for all scenarios. The highest biases were observed in outcome scenario VI and missingness scenario (v) ($\abs{RB}$ $>$ 100\%).

\paragraph{Low exposure prevalence (10\%) scenario.} Generally, the performance of the missing data methods in the low exposure prevalence (10\%) scenario was worse than in the 50\% exposure scenario. CCA was still unbiased in m-DAGs A and B ($\abs{RB}$ $<$ 5\%), slightly biased in m-DAGs D and E for outcome scenarios I-V ($\abs{RB}$ $<$ 9\%, and $\abs{RB}$ 21-33\% in outcome scenario VI), and highly biased in all other scenarios ($\abs{RB}$ $>$ 58\% in m-DAG C, F-J). MI approaches were more biased in most scenarios for m-DAGs A, B, D, E, F and H with 10\% exposure prevalence compared with the 50\% exposure prevalence scenarios, except for MI-SMC and MI-Com which continued to be approximately unbiased or only slightly biased in m-DAGs A, D and E ($\abs{RB}$ $<$ 15\%). MI approaches were extremely biased in missingness scenarios (ii-iv) for m-DAG C ($\abs{RB}$ 18-99\%). For m-DAGs G and J, all methods were severely biased ($\abs{RB}$ $>$ 62\%). Finally, under m-DAG I, the MI approaches showed biases of varying extent ($\abs{RB}$ ranging from -56\% to 44\%) except for missingness scenario (v) where all approaches were biased ($\abs{RB}$ $>$ 45\%). 

\paragraph{Other performance indicators.} The MCSE for bias ranged from 0.0023 to 0.0036 across m-DAGs in 50\% exposure prevalence and from 0.0025 to 0.007 for 10\% exposure. Figure \ref{EmpSE} shows the empirical SEs for each method. These were generally similar across all the m-DAGs and scenarios for each missing data method in 50\% exposure (0.10-0.16) and were larger in 10\% exposure (0.12-0.47). The CCA yielded the highest empirical SE (0.28-0.47) in m-DAG E for the low exposure prevalence scenario. Comparing the estimated SEs to the empirical SEs, they were approximately the same in the 50\% exposure prevalence scenario (relative error $<$ 7\%, 0.75 percentile was 2\%) and relatively close to each other in the 10\% exposure prevalence scenario (relative error $<$ 32\%, 0.75 percentile was 9.7\%). Figure \ref{Coverage} shows the coverage across the scenarios. For both exposure prevalence scenarios, the coverage probabilities of the 95\% CI given by missing data methods were close to the nominal coverage where there was minimal bias in the point ACE estimate, although some under-coverage was observed in the 10\% exposure prevalence setting, which could be explained by the larger absolute value of estimated SE compared with the empirical SE. Under-coverage was observed when there was bias in the point ACE estimate.

\section{Application to the VAHCS case study}
We first developed an m-DAG for the case study based on substantive knowledge. See Table \ref{VAHCS.assumption}. The closest canonical m-DAG is m-DAG J if all ``likely’’ arrows are strictly present. Under this assumption, our theoretical results show that the ACE is non-recoverable (whether we assume the presence or absence of $\bm{W}$) and ACE estimates obtained using each of the missing data methods, provided in Table \ref{case}, are likely to be biased based on the simulation results. In this setting, researchers need to consider sensitivity analysis methods based on external information, e.g. delta-adjustment methods \citep{leacy2016multiple,moreno2016sensitivity,tompsett2018use}. However, if some of the arrows were missing, these estimates could be interpreted with more confidence, suggesting a moderately negative effect of cannabis use in female adolescents on their mental health in young adulthood (Table \ref{case}).

\section{Discussion}
We investigated the recoverability of the ACE in point-exposure epidemiological studies under typical missingness mechanisms depicted by m-DAGs in the context of effect modification, extending the results in \citep{moreno2018canonical}. Our theoretical results found the ACE to be recoverable in m-DAGs where no variable causes its own missingness (m-DAGs A-C) and non-recoverable in other m-DAGs (m-DAGs G and H conjectured). Furthermore, our simulation study showed that when estimating the ACE using g computation the CCA was approximately unbiased when the outcome did not influence missingness in any variable and was severely biased in all other missingness mechanisms. The ACE could be estimated by all of the considered MI approaches with low bias in m-DAGs where the outcome causes its own missingness or it causes the missingness of a variable that causes its own missingness if the exposure prevalence is moderate (50\%). There was no noticeable difference among MI approaches except where there were strong exposure-incomplete confounder interactions, in which MI-SMC and MI-Com approaches were the least biased across scenarios assessed. With low exposure prevalence, MI approaches performed less well.

\subsection{Theoretical work}
In the context of effect modification, the recoverability of the ACE is inferred to rely on the recoverability of the joint distribution of confounders as it is one of the key factors in the identification formula in the absence of missing data (\ref{ACE}). Given the non-recoverability of the conditional outcome distribution in m-DAGs G and H, we conjectured that the ACE is non-recoverable for these m-DAGs but did not provide a formal proof. Proving non-recoverability is not straightforward because, generally, the non-recoverability of a component in a given decomposition of a target distribution is not a sufficient condition for non-recoverability of the target distribution. Other decompositions might exist for the target distribution such that each component is recoverable, as arises with inverse probability weighting-type expressions \citep{mohan2013graphical,mohan2014graphical,mohan2021graphical,bhattacharya2020identification}. As a result, even automatised identification algorithms, such as that implemented in the R package {\tt{dosearch}} developed by \citep{tikka2019causal}, are not ``complete’’, meaning they cannot establish (non-)recoverability for all problems. Of note, as a post-hoc check, we applied the {\tt{dosearch}} algorithm to our problems, which also showed that the ACE was recoverable in m-DAGs A-C and did not reach a solution for the remaining m-DAGs.

In our proofs, we interpreted an arrow from $\bm{Z_2}$ to $M_{\bm{Z_2}}$ to mean that there is an arrow from at least one incomplete confounder to its own missingness. In the case where missingness in an incomplete confounder is caused by other incomplete confounders but not itself, the same m-DAGs apply, while the recoverability results would differ. Specifically for m-DAG D, the ACE would be recoverable because the joint distribution of the confounders would be recoverable (proof shown in Appendix). We conjecture that the ACE would still be non-recoverable in this setting in m-DAGs E, F and I because the exposure distribution remains non-recoverable. 

\subsection{Simulation study}
In the simulation study, we considered more general m-DAGs, where unmeasured common causes for the missingness indicators were included to capture the dependency amongst these indicators. By depicting missingness mechanisms using m-DAGs, our simulation study allowed us to assess the performance of different missing data methods beyond the MCAR/MAR/MNAR framework, enabling us to develop practical guidance for the estimation of the ACE in settings with multivariable missing data. The theoretical non-recoverability of ACE in m-DAGs D-J is attributed to incomplete variables causing their own missingness. However, if the outcome does not cause missingness in any variable (m-DAGs D and E), we found that all considered methods yielded approximately unbiased estimates. Otherwise, none of the missing data methods achieved unbiased estimation (m-DAGs F-J). Based on our simulation results, we conclude that sensitivity analysis approaches based on external information are necessary if the outcome is likely to cause its own missingness or missingness in any other incomplete variable that causes its own missingness. 

However, for the missingness mechanism where the sensitivity analysis is not required (m-DAGs A-E), our simulation study also showed the influence of many factors on the performance of the MI approaches: exposure prevalence, the strength of the interactions in the outcome, exposure-confounder interaction on missingness. For example, the MI-based methods were approximately unbiased when estimating the ACE in m-DAG C in the 50\% exposure prevalence, but they were biased with 10\% exposure prevalence. It is possible that lower exposure prevalence would result in random violations of the positivity assumption and/or the larger imputation variance. The latter might also explain the bias in MI approaches for m-DAGs A and B. All missing data methods performed worse when there was a strong interaction between exposure and incomplete confounder in the outcome model (larger biases in outcome scenario VI) as well as when the missingness mechanisms included an interaction between exposure and incomplete variables (larger biases in missingness scenarios (iii)-(v)). 

In settings with no interactions in the outcome model, our findings for settings with no interactions in the outcome model were consistent with findings in \citep{moreno2018canonical}, where a simulation study was used to evaluate the performance of CCA and MI in estimating the ACE in the absence of effect modification. In settings with interactions in the outcome model, our finding showed that ensuring compatibility between the imputation and outcome model in g-computation was key to the performance of MI. For the MI-SMC approach, using the outcome model of g-computation (plus auxiliary variables) as the substantive model provided the least biased approach across all scenarios. For the MI-Com approach, the approximately compatible imputation model can give unbiased estimation by including interactions in the way proposed by \citep{tilling2016appropriate}.

\subsection{Strength and limitations}
Our work focused on ten canonical m-DAGs that had been previously constructed to be representative of typical missingness mechanisms in epidemiological studies. As such, our work enabled us to provide practical guidance in the presence of multivariable missing data. However, these m-DAGs employ some simplifications, such as treating all incomplete confounders as a single vector, which ignores subtleties regarding the interpretation of the arrows and ultimately recoverability, as previously mentioned. In the simulation study, we assessed a continuous outcome aligned with VHACS data. But further research simulating a binary outcome would be valuable for assessing methods in the context of non-linear models. Our findings are restricted to the context of a correctly specified analysis model, which may not be the case in practice. It would be of interest for further research to consider the performance of CCA and MI in the context of misspecified analysis models. Additionally, research considering other causal effect estimation methods would be important, furthering existing work on handling missing data with TMLE and the propensity score method \citep{blake2020propensity,dashti2021handling}. Regarding implementing MI approaches, our study considered an auxiliary variable to improve the precision of MI but not one that impacted recoverability in the m-DAGs. Further research could extend recoverability results to canonical m-DAGs that include auxiliary variables. Lastly, the number of imputations was low in our simulation due to computational cost when applying bootstrapping after MI, which may have affected the performance of MI in low exposure prevalence settings. We pooled the ACE estimates and SEs from bootstrapping by applying Rubin’s rules to obtain an overall MI estimate and its SE for the MI approaches. Such a procedure is referred to as MI bootstrap Rubin in \citep{bartlett2020bootstrap}. It has been shown that MI bootstrap Rubin yields confidence intervals with approximately nominal coverage when the imputation and analysis procedures are congenial and correctly specified \citep{bartlett2020bootstrap}. An alternative is to use the von Hippel bootstrapped MI approach, which is recommended to achieve nominal coverage in uncongenial and/or misspecified scenarios \citep{von2021maximum,bartlett2020bootstrap}.

\subsection{Conclusion}
In summary, we investigated the recoverability and unbiased estimation of the ACE under the missingness assumptions depicted by m-DAGs, an intuitive and transparent framework to communicate and assess missingness assumptions. In m-DAGs where ACE can be unbiasedly estimated, we provided recommendations for implementing MI with g-computation based on achieving compatibility with the outcome model. In other m-DAGs, we concluded the necessity of sensitivity analysis. We demonstrated how missingness assumptions were applied to interpret the analysis results by the VAHCS case study.

\section{DECLARATIONS}

\noindent {\bf{Supplementary data}}\\
Supplementary data are available online.
\\
\noindent {\bf{Ethics approval}}\\
The case study used data from the Victorian Adolescent Health Cohort Study. Data collection protocols were approved by The Royal Children’s Hospital’s Ethics in Human Research Committee. Informed parental consent was obtained for each participant prior to entry.
\\
\noindent {\bf{Funding}}\\
This work was supported by funding from the Australian National Health and Medical Research Council (Project Grant 1166023, Career Development Fellowship 1127984 to KJL, Investigator Grant 2009572 to MMB). JZ is funded by the Melbourne Research Scholarship and a top-up scholarship from Statistical Society of Australia. The Murdoch Children’s Research Institute is supported by the Victorian Government’s Operational Infrastructure Support Program. The funding bodies do not have any role in the collection, analysis, interpretation or writing of the study.\\
\noindent {\bf{Data availability}}\\
The data underlying this article will be shared on reasonable request to the corresponding author with the permission of Professor George Patton. The data for simulation study are openly available at \\
{\tt{https://github.com/Jiaxin-Zhang-GitHub/Recoverability-and-estimation.git}} 
\\
\noindent {\bf{Acknowledgements}}\\
The authors would like to thank the  Victorian Centre for Biostatistics (ViCBiostat) Causal Inference group, Missing Data group and other members of ViCBiostat for providing feedback in designing and interpreting the simulation study. 
\\
\noindent {\bf{Author contributions}}\\
JZ, SGD, JBC, KJL and MMB conceived the project and designed the study. JZ completed the theoretical proof and conducted the simulation study and data analysis, with input from co-authors, and drafted the manuscript. MMB, SGD, JBC and KJL provided critical input to the manuscript. 
\\
\noindent {\bf{Conflict of interest}}\\
None declared

\section*{Appendix {\it(Recoverability of average causal effect in canonical m-DAGs)}}

\subsection*{A.1.\enspace Notations and assumptions}\label{assumption}
For substantive variables, we use $X$, $Y$, $\bm{Z_1}$, $\bm{Z_2}$ and $\bm{U}$ to represent the exposure, outcome, vectors of complete and incomplete confounders, and the vector of unknown common causes of exposure and confounders, respectively. The missingness indicator of a incomplete variable is denoted by $M$ with a corresponding subscript (1 if missing and 0 otherwise). We use $\bm{m}=\bm{0}$ as a shorthand for $\{m_X=m_Y=m_{\bm{Z_2}}=0\}$. The causal quantity of interest is the ACE of $X$ on $Y$. To prove its recoverability it is sufficient to prove the recoverability of the marginal distribution of the potential outcome under a given exposure value. In the framework of do-calculus, such a marginal distribution is expressed by ${\rm P}(y|do(x))$ or equivalently ${\rm P}(y|\hat{x})$, where the $do(x)$ operator represents the intervention taking the value $x$. We provide the proof of the recoverability of the ACE in the m-DAGs of Figure \ref{m-DAG} (in the main text) under the assumption that there are no unmeasured common causes for missingness indicators ($\bm{W}$), that is, that the dashed arrows in Figure \ref{m-DAG} are absent. We assumed categorical variables in these results, but the conclusion can easily be extended to continuous variables. 

\subsection*{A.2.\enspace Definitions and established results}
The following theoretical results on recoverability and non-recoverability concepts were developed by \citep{mohan2013graphical, mohan2014graphical}.

\begin{definition}[Recoverability, \citep{mohan2013graphical}]\label{rec}
Given an m-graph $G$, and a target relation $Q$ defined on the variables in the set of observable variables $V$, $Q$ is said to be recoverable in $G$ if there exists an algorithm that produces a consistent estimate of $Q$ for every dataset $D$ such that ${\rm P}(D)$ is (1) compatible with $G$ and (2) strictly positive over complete cases.
\end{definition}

\begin{corollary}[\citep{mohan2013graphical}]
A relation $Q$ is recoverable in $G$ if and only if $Q$ can be expressed in terms of the probability ${\rm P}(O)$ where $O$ is the set of observable variables in $G$. In other words, for any two models $M_1$ and $M_2$ inducing distributions ${\rm P}^{M_1}$ and ${\rm P}^{M_1}$ respectively, if ${\rm P}^{M_1}(O) = {\rm P}^{M_1}(O)>0$ then $Q^{M_1} =Q^{M_2}$.
\end{corollary}

\begin{theorem}[Conditions for recoverability, \citep{mohan2013graphical}]
A query $Q$ defined over any variable in $V$ is recoverable if it is decomposable into terms of conditional distributions such that every partially observed variable appears only in conditional distributions that condition on their own missingness indicator being zero (i.e. not missing).
\end{theorem}

\begin{corollary}\label{non-rec}
A relation $Q$ is non-recoverable in $G$ if either of the following two statements stands:\\
(1) The relation $Q$ can be expressed as the sum of a recoverable distribution plus a non-recoverable distribution.\\
(2) The relation $Q$ can be expressed as the product of a recoverable distribution multiplied by a non-recoverable distribution, if the recoverable distribution is strictly positive.
\end{corollary}
\begin{proof}
Let $Q_R$ and $Q_N$ be the recoverable and non-recoverable distributions in $G$, respectively. There exist two models $M_1$ and $M_2$ that are compatible with $G$ such that:
\begin{equation*}
\begin{aligned}
{\rm P}^{M_1}(O)&={\rm P}^{M_2}(O)>0,\\
Q_R^{M_1}&=Q_R^{M_2},\\
Q_N^{M_1}&\neq Q_N^{M_2},
\end{aligned}
\end{equation*}
If $Q$ can be expressed by $Q_R+Q_N$, then
\begin{equation*}
Q^{M_1}=Q_R^{M_1}+Q_N^{M_1}\neq Q_R^{M_2}+Q_N^{M_2}=Q^{M_2}.
\end{equation*}
If $Q$ can be expressed by $Q_R Q_N$, then
\begin{equation*}
Q^{M_1}=Q_R^{M_1} Q_N^{M_1}\neq Q_R^{M_2} Q_N^{M_2}=Q^{M_2},
\end{equation*}
given $Q_R$ is strictly positive. Thus, $Q$ is non-recoverable in both cases.
\end{proof}

\begin{theorem}[Non-recoverability of joint distribution, \citep{mohan2014graphical}]\label{non-rec-V}
Given m-DAG $G$, the following conditions are necessary for recoverability of a joint distribution:\\
(1) Any incomplete variable and its missingness indicator are not neighbours, i.e. not connected by arrows, and \\
(2) There does not exist a path any incomplete variable to its missingness indicator in which every intermediate node is both a collider and a substantive variable.
\end{theorem}

\begin{corollary}[Non-recoverability of conditional distribution, \citep{mohan2014graphical}]\label{non-rec-C}
A conditional distribution is non-recoverable in a given m-DAG $G$ if one of the following conditions is true:\\
(1) The variable being conditioned upon and its missingness indicator are neighbours, i.e. connected by arrows. \\
(2) $G$ contains a collider path connecting the variable being conditioned upon and its missingness indicator such that all intermediate nodes of the path are conditioned.
\end{corollary}

\begin{lemma}[\citep{mohan2013graphical,tian2010identifying}]\label{G.G'}
If a target relation $Q$ is not recoverable in m-graph $G$, then $Q$ is not recoverable in the graph $G'$ resulting from adding a single edge to $G$.
\end{lemma}

\subsection*{A.3.\enspace do-calculus and probability results}

Below is an introduction to the rules of do-calculus \citep{pearl2009causality}. Denote the graph after removing all arrows pointing to $x$ as $G_{\overline{X}}$, and the graph after removing all arrows emerging from $x$ as $G_{\underline{X}}$.\\
Rule-1 (Insertion/deletion of observations):
\begin{equation*}
    {\rm P}(y|\hat{x},z,m)={\rm P}(y|\hat{x},m) \text{, if } Y\ind Z|X,M \text{ in } G_{\overline{X}}
\end{equation*}
Rule-2 (Action/observation exchange):
\begin{equation*}
    {\rm P}(y|\hat{x},\hat{z},m)={\rm P}(y|\hat{x},z,m) \text{, if } Y\ind Z|X,M \text{ in } G_{\overline{X},\underline{Z}}
\end{equation*}
Rule-3 (Insertion/deletion of actions):
\begin{equation*}
    {\rm P}(y|\hat{x},\hat{z},m)={\rm P}(y|\hat{x},m) \text{, if } Y\ind Z|X,M \text{ in } G_{\overline{X},\overline{Z(M)}}
\end{equation*}
where $Z(M)$ is the set of $Z$-nodes that are not ancestors of any $M$-node in $G_{\overline{X}}$.

The following two probability results are widely used in the proof.  
\begin{equation}\label{A.eq:1}
\begin{aligned}
    {\rm P}(A=a|B)&=\frac{{\rm P}(A=a,B=b)}{{\rm P}(B=b)}\frac{{\rm P}(A=a,B=b,C=c)}{{\rm P}(A=a,B=b,C=c)}\\
    &=\frac{{\rm P}(A=a,C=c|B)}{{\rm P}(C=c|A,B)}\\
    &=\frac{{\rm P}(A=a|B,C=c){\rm P}(C=c|B)}{{\rm P}(C=c|A,B)}
\end{aligned}
\end{equation}
\begin{equation}\label{A.eq:2}
\begin{aligned}
    {\rm P}(A=a|B)&=\frac{{\rm P}(A=a,B=b)}{{\rm P}(B=b)}\\
    &=\frac{{\rm P}(A=a,B=b)}{\sum_C {\rm P}(B=b,C=c)}\\
    &=\frac{{\rm P}(A=a,B=b)}{\sum_C {\rm P}(B=b,C=c)\frac{{\rm P}(A=a,B=b,C=c)}{{\rm P}(A=a,B=b,C=c)}}\\
    &=\frac{{\rm P}(A=a,B=b)}{\sum_C\frac{{\rm P}(A=a,B=b,C=c)}{{\rm P}(A=a|B,C)}}
\end{aligned}
\end{equation}

\subsection*{A.4.\enspace Recoverability for ACE in m-DAG B}\label{sec.B}
The ACE is recoverable in m-DAG B shown in Figure \ref{m-DAG.B}. 

\begin{proof}
The marginal distribution of the potential outcome ${\rm P}(y|\hat{x})$ is given by
\begin{equation}\label{A.B.1}
    {\rm P}(y|\hat{x})=\sum_{\bm{z_1}}\sum_{\bm{z_2}}{\rm P}(y|\hat{x},\bm{z_1},\bm{z_2}){\rm P}(\bm{z_1},\bm{z_2}|\hat{x}).
\end{equation}

Following Rule-1, we can condition the first term on all missingness indicators:
\begin{equation*}
\begin{aligned}
    {\rm P}(y|\hat{x},\bm{z_1},\bm{z_2})
    &={\rm P}(y|\hat{x},\bm{z_1},\bm{z_2},m_X=0,m_Y=0,m_{\bm{Z_2}}=0)\\
    &={\rm P}(y|\hat{x},\bm{z_1},\bm{z_2},\bm{m}=\bm{0}),
\end{aligned}
\end{equation*}
as $Y\ind M_Y|X,\bm{Z_1},\bm{Z_2}$ in $G_{\overline{X}}$, $Y\ind M_X|X,\bm{Z_1},\bm{Z_2},M_Y$ in $G_{\overline{X}}$, and $Y\ind M_{\bm{Z_2}}|X,\bm{Z_1},\bm{Z_2},M_Y,M_X$ in $G_{\overline{X}}$. Applying Rule-2, 
\begin{equation*}
    {\rm P}(y|\hat{x},\bm{z_1},\bm{z_2},\bm{m}=\bm{0})={\rm P}(y|x,\bm{z_1},\bm{z_2},\bm{m}=\bm{0}),
\end{equation*}
since $Y\ind X|\bm{Z_1},\bm{Z_2},M_Y,M_X,M_{\bm{Z_2}}$ in $G_{\underline{X}}$. Hence, the first term of (\ref{A.B.1}) is recoverable.

The second term of (\ref{A.B.1}) can be simplified as
\begin{equation*}
    {\rm P}(\bm{z_1},\bm{z_2}|\hat{x})={\rm P}(\bm{z_1},\bm{z_2})={\rm P}(\bm{z_1}){\rm P}(\bm{z_2}|\bm{z_1}),
\end{equation*}
by Rule-3, given $(\bm{Z_1},\bm{Z_2})\ind X$ in $G_{\overline{X}}$. Therefore we can express ${\rm P}(\bm{z_2}|\bm{z_1})$ with the use of equation (\ref{A.eq:1}) as
\begin{equation*}
    {\rm P}(\bm{z_2}|\bm{z_1})=\frac{{\rm P}(\bm{z_2}|\bm{z_1},m_X=0,m_{\bm{Z_2}}=0){\rm P}(m_X=0,m_{\bm{Z_2}}=0|\bm{z_1})}{{\rm P}(m_X=0,m_{\bm{Z_2}}=0|\bm{z_1},\bm{z_2})},
\end{equation*}
where both numerators are recoverable. Rewrite the denominator by equation (\ref{A.eq:2}) as
\begin{equation*}
    {\rm P}(m_X=0,m_{\bm{Z_2}}=0|\bm{z_1},\bm{z_2})=\frac{{\rm P}(m_X=0,m_{\bm{Z_2}}=0,\bm{z_1},\bm{z_2})}{\sum_{x'}\frac{{\rm P}(m_X=0,m_{\bm{Z_2}}=0,x',\bm{z_1},\bm{z_2})}{{\rm P}(m_X=0,m_{\bm{Z_2}}=0|x',\bm{z_1},\bm{z_2})}},
\end{equation*}
where ${\rm P}(m_X=0,m_{\bm{Z_2}}=0,\bm{z_1},\bm{z_2})$ and ${\rm P}(m_X=0,m_{\bm{Z_2}}=0,x',\bm{z_1},\bm{z_2})$ are recoverable. Now, 
\begin{equation*}
    {\rm P}(m_X=0,m_{\bm{Z_2}}=0|x',\bm{z_1},\bm{z_2})={\rm P}(m_{\bm{Z_2}}=0|m_X=0,x',\bm{z_1},\bm{z_2}){\rm P}(m_X=0|x',\bm{z_1},\bm{z_2}),
\end{equation*}
where the first factor is recoverable. This is because $M_{\bm{Z_2}}\ind \bm{Z_2}|X,\bm{Z_1},M_X$ in $G$. Therefore, by Rule-1,
\begin{equation*}
    {\rm P}(m_{\bm{Z_2}}=0|m_X=0,x',\bm{z_1},\bm{z_2})={\rm P}(m_{\bm{Z_2}}=0|m_X=0,x',\bm{z_1}).
\end{equation*}
The second factor is also recoverable, because $M_X\ind M_{\bm{Z_2}}|X,\bm{Z_1},\bm{Z_2}$ in $G$. Thus, by Rule-1,
\begin{equation*}
    {\rm P}(m_X=0|x',\bm{z_1},\bm{z_2})={\rm P}(m_X=0|x',\bm{z_1},\bm{z_2},m_{\bm{Z_2}}=0),
\end{equation*}
and $M_X\ind X|\bm{Z_1},\bm{Z_2},M_{\bm{Z_2}}$ in $G$. Thus, by Rule-1,
\begin{equation*}
    {\rm P}(m_X=0|x',\bm{z_1},\bm{z_2},m_{\bm{Z_2}}=0)={\rm P}(m_X=0|\bm{z_1},\bm{z_2},m_{\bm{Z_2}}=0).
\end{equation*}
Therefore, ${\rm P}(\bm{z_2}|\bm{z_1})$ can be expressed as below and is recoverable, as every factor is recoverable:
\begin{equation*}
\begin{aligned}
   {\rm P}(\bm{z_2}|\bm{z_1})&=\frac{{\rm P}(\bm{z_2}|\bm{z_1},m_X=0,m_{\bm{Z_2}}=0){\rm P}(m_X=0,m_{\bm{Z_2}}=0|\bm{z_1})}{{\rm P}(m_X=0,m_{\bm{Z_2}}=0,\bm{z_1},\bm{z_2})}\times \\
&\sum_{x'}\frac{{\rm P}(m_X=0,m_{\bm{Z_2}}=0,x',\bm{z_1},\bm{z_2})}{{\rm P}(m_{\bm{Z_2}}=0|m_X=0,x',\bm{z_1}){\rm P}(m_X=0|\bm{z_1},\bm{z_2},m_{\bm{Z_2}}=0)}.
\end{aligned}
\end{equation*}
Hence, the marginal distribution of the potential outcome ${\rm P}(y|\hat{x})$ is recoverable, by the following expression
\begin{equation*}
\begin{aligned}
   {\rm P}(y|\hat{x})&=\sum_{\bm{z_1}}\sum_{\bm{z_2}}{\rm P}(y|x,\bm{m}=\bm{0},\bm{z_1},\bm{z_2}){\rm P}(\bm{z_1})\times \\
   &\frac{{\rm P}(\bm{z_2}|\bm{z_1},m_X=0,m_{\bm{Z_2}}=0){\rm P}(m_X=0,m_{\bm{Z_2}}=0|\bm{z_1})}{{\rm P}(m_X=0,m_{\bm{Z_2}}=0,\bm{z_1},\bm{z_2})}\times \\
&\sum_{x'}\frac{{\rm P}(m_X=0,m_{\bm{Z_2}}=0,x',\bm{z_1},\bm{z_2})}{{\rm P}(m_{\bm{Z_2}}=0|m_X=0,x',\bm{z_1}){\rm P}(m_X=0|\bm{z_1},\bm{z_2},m_{\bm{Z_2}}=0)}\\
&=\sum_{\bm{z_1}}\sum_{\bm{z_2}}{\rm P}(y|x,\bm{z_1},\bm{z_2},\bm{m}=\bm{0})\times \\
&\sum_{x'}\frac{{\rm P}(m_X=0,m_{\bm{Z_2}}=0,x',\bm{z_1},\bm{z_2})}{{\rm P}(m_{\bm{Z_2}}=0|x',\bm{z_1},m_X=0){\rm P}(m_X=0|\bm{z_1},\bm{z_2},m_{\bm{Z_2}}=0)}.
\end{aligned}    
\end{equation*}
Note the marginal distribution of the potential outcome is recoverable in m-DAG B can be derived from the fact that the potential outcome is recoverable in m-DAG C based on Lemma \ref{G.G'}. The proof for the latter is shown as below.

\end{proof}

\subsection*{A.5.\enspace Recoverability for ACE in m-DAG C}
The ACE is recoverable in m-DAG C shown in Figure \ref{m-DAG.C}.

\begin{proof}
As shown in equation (\ref{A.B.1}), it is possible to express the marginal distribution of the potential outcome as
\begin{equation*}
\begin{aligned}
   {\rm P}(y|\hat{x})&=\sum_{\bm{z_1}}\sum_{\bm{z_2}}{\rm P}(y|\hat{x},\bm{z_1},\bm{z_2}){\rm P}(\bm{z_1},\bm{z_2}|\hat{x})\\
   &=\sum_{\bm{z_1}}\sum_{\bm{z_2}}{\rm P}(y|\hat{x},\bm{z_1},\bm{z_2},m_Y=0){\rm P}(\bm{z_1},\bm{z_2}),
\end{aligned}
\end{equation*}
where the first term is derived by Rule-1, given that $Y\ind M_Y|X,\bm{Z_1},\bm{Z_2}$ in $G_{\overline{X}}$, and the second term is derived by Rule-3, given that  $(\bm{Z_1},\bm{Z_2}) \ind X$ in $G_{\overline{X}}$. We can rewrite the first term as 

\begin{equation}
\begin{aligned}
    &{\rm P}(y|\hat{x},\bm{z_1},\bm{z_2},m_Y=0)\\
    &=\frac{{\rm P}(y|\hat{x},\bm{z_1},\bm{z_2},m_X=m_Y=m_{\bm{Z_2}}=0){\rm P}(m_X=m_{\bm{Z_2}}=0|\hat{x},\bm{z_1},\bm{z_2},m_Y=0)}{{\rm P}(m_X=m_{\bm{Z_2}}=0|\hat{x},y,\bm{z_1},\bm{z_2},m_Y=0)}\\
&=\frac{{\rm P}(y|\hat{x},\bm{z_1},\bm{z_2},m_X=m_Y=m_{\bm{Z_2}}=0){\rm P}(m_X=m_{\bm{Z_2}}=0|\hat{x},\bm{z_1},\bm{z_2},m_Y=0)}{{\rm P}(m_{\bm{Z_2}}=0|\hat{x},y,\bm{z_1},\bm{z_2},m_X=m_Y=0){\rm P}(m_X=0|\hat{x},y,\bm{z_1},\bm{z_2},m_Y=0)},
\end{aligned}
\end{equation}
where the first line used equation (\ref{A.eq:1}). Now we assess the recoverability of four factors in this fraction.

Applying Rule-2 to the first term in the numerator, given that $Y\ind X|\bm{Z_1},\bm{Z_2},M_X,M_Y,M_{\bm{Z_2}}$ in $G_{\underline{X}}$ shows that, 
\begin{equation*}
    {\rm P}(y|\hat{x},\bm{z_1},\bm{z_2},m_X=m_Y=m_{\bm{Z_2}}=0)={\rm P}(y|x,\bm{z_1},\bm{z_2},\bm{m}=\bm{0}),
\end{equation*}
which is recoverable.

As for the second term, 
\begin{equation*}
\begin{aligned}
    {\rm P}(m_X=m_{\bm{Z_2}}=0|\hat{x},\bm{z_1},\bm{z_2},m_Y=0)&={\rm P}(m_X=m_{\bm{Z_2}}=0|\hat{x},\bm{z_1},\bm{z_2})\\
&={\rm P}(m_X=m_{\bm{Z_2}}=0|x,\bm{z_1},\bm{z_2}).
\end{aligned}
\end{equation*}
The first line is derived from Rule-1, given that $(M_{\bm{Z_2}},M_X)\ind M_Y|\bm{Z_1},\bm{Z_2},X$ in $G_{\overline{X}}$, and the second line is derived from Rule-2, given that $(M_{\bm{Z_2}},M_X)\ind X|\bm{Z_1},\bm{Z_2}$ in $G_{\underline{X}}$. 

Next, 
\begin{equation*}
\begin{aligned}
    {\rm P}(m_{\bm{Z_2}}=0|\hat{x},y,\bm{z_1},\bm{z_2},m_X=m_Y=0)&={\rm P}(m_{\bm{Z_2}}=0|\hat{x},y,\bm{z_1},m_X=m_Y=0)\\
&={\rm P}(m_{\bm{Z_2}}=0|x,y,\bm{z_1},m_X=m_Y=0)
\end{aligned}    
\end{equation*}
is recoverable, where the first line is derived from Rule-1, given that $M_{\bm{Z_2}}\ind \bm{Z_2}|(X,Y,\bm{Z_1},M_Y,M_X)$ in $G_{\overline{X}}$, and the second line is derived from Rule-2, given that $M_{\bm{Z_2}}\ind X|(\bm{Z_1},Y,M_X,M_Y)$ in $G_{\underline{X}}$.

Lastly, 
\begin{equation*}
\begin{aligned}
{\rm P}(m_X=0|\hat{x},y,\bm{z_1},\bm{z_2},m_Y=0)&={\rm P}(m_X=0|\hat{x},y,\bm{z_1},\bm{z_2},m_Y=m_{\bm{Z_2}}=0)\\
&={\rm P}(m_X=0|y,\bm{z_1},\bm{z_2},m_Y=m_{\bm{Z_2}}=0)  
\end{aligned}    
\end{equation*}
is also recoverable. The first line is derived from Rule-1, given that $M_X \ind M_{\bm{Z_2}}|(X,Y,\bm{Z_1},\bm{Z_2},M_Y)$ in $G_{\overline{X}}$, and the second line is derived from Rule-3, given that $M_X \ind X|(\bm{Z_1},\bm{Z_2},Y,M_{\bm{Z_2}},M_Y)$ in $G$.
Using these expressions we can rewrite the distribution of the potential outcome as
\begin{equation*}
\begin{aligned}
&{\rm P}(y|\hat{x})\\&=\sum_{\bm{z_1}}\sum_{\bm{z_2}}\frac{{\rm P}(y|x,\bm{z_1},\bm{z_2},\bm{m}=\bm{0}){\rm P}(m_X=m_{\bm{Z_2}}=0|x,\bm{z_1},\bm{z_2}){\rm P}(\bm{z_1},\bm{z_2})}{{\rm P}(m_{\bm{Z_2}}=0|x,y,\bm{z_1},m_X=m_Y=0){\rm P}(m_X=0|y,\bm{z_1},\bm{z_2},m_Y=m_{\bm{Z_2}}=0)} \\
&=\sum_{\bm{z_1}}\sum_{\bm{z_2}}\frac{{\rm P}(y|x,\bm{z_1},\bm{z_2},\bm{m}=\bm{0}){\rm P}(x,\bm{z_1},\bm{z_2},m_X=m_{\bm{Z_2}}=0)}{{\rm P}(m_{\bm{Z_2}}=0|x,y,\bm{z_1},m_X=m_Y=0){\rm P}(m_X=0|y,\bm{z_1},\bm{z_2},m_Y=m_{\bm{Z_2}}=0){\rm P}(x|\bm{z_1},\bm{z_2})},
\end{aligned}    
\end{equation*}
where the second equality arises from Bayes' Theorem.

Thus, the marginal distribution of the potential outcome is recoverable because every factor is recoverable, where the conditional distribution of ${\rm P}(x|\bm{z_1},\bm{z_2})$ is recoverable in m-DAG C by Corollary \ref{cor.3} below.
\end{proof}

\begin{lemma}\label{lemma.1}
In m-DAG C, ${\rm P}(\bm{m}=\bm{0}|x,y,\bm{z_1},\bm{z_2})$ is recoverable.
\end{lemma}
\begin{proof}
Expand the distribution as
\begin{equation*}
\begin{aligned}
&{\rm P}(\bm{m}=\bm{0}|x,y,\bm{z_1},\bm{z_2})={\rm P}(m_Y=0|x,y,\bm{z_1},\bm{z_2},m_X=m_{\bm{Z_2}}=0)\times\\ 
&{\rm P}(m_X=0|x,y,\bm{z_1},\bm{z_2},m_{\bm{Z_2}}=0){\rm P}(m_{\bm{Z_2}}=0|x,y,\bm{z_1},\bm{z_2}).    
\end{aligned}    
\end{equation*}
We now show each factor is recoverable in m-DAG C. First, given that $M_Y\ind Y|X,\bm{Z_1},\bm{Z_2},M_X,M_{\bm{Z_2}}$ in $G$, by Rule-1,
\begin{equation*}
    {\rm P}(m_Y=0|x,y,\bm{z_1},m_X=m_{\bm{Z_2}}=0)={\rm P}(m_Y=0|x,\bm{z_1},\bm{z_2},m_X=m_{\bm{Z_2}}=0).
\end{equation*}
Second, given that $M_X\ind M_Y|X,\bm{Z_1},\bm{Z_2},Y,M_{\bm{Z_2}}$ in $G$, and $M_X\ind X|\bm{Z_1},\bm{Z_2},Y,M_{\bm{Z_2}},M_Y$ in $G$,
\begin{equation*}
\begin{aligned}
    {\rm P}(m_X=0|x,y,\bm{z_1},\bm{z_2},m_{\bm{Z_2}}=0)&={\rm P}(m_X=0|x,y,\bm{z_1},\bm{z_2},m_Y=m_{\bm{Z_2}}=0)\\
    &={\rm P}(m_X=0|y,\bm{z_1},\bm{z_2},m_Y=m_{\bm{Z_2}}=0)
\end{aligned}    
\end{equation*}
is recoverable by Rule-1. Third, given that $M_{\bm{Z_2}}\ind(M_X,M_Y)|X,\bm{Z_1},\bm{Z_2},Y$ in $G$, and $M_{\bm{Z_2}}\ind \bm{Z_2}|X,\bm{Z_1},Y,M_X,M_Y$ in $G$,
\begin{equation*}
\begin{aligned}
    {\rm P}(m_{\bm{Z_2}}=0|x,y,\bm{z_1},\bm{z_2})&={\rm P}(m_{\bm{Z_2}}=0|x,y,\bm{z_1},\bm{z_2},m_X=m_Y=0)\\
    &={\rm P}(m_{\bm{Z_2}}=0|x,y,\bm{z_1},m_X=m_Y=0)
\end{aligned}    
\end{equation*}
is recoverable by Rule-1. Hence, ${\rm P}(\bm{m}=\bm{0}|x,y,\bm{z_1},\bm{z_2})$ is recoverable and given by
\begin{equation*}
\begin{aligned}
    &{\rm P}(\bm{m}=\bm{0}|x,y,\bm{z_1},\bm{z_2})={\rm P}(m_Y=0|x,\bm{z_1},\bm{z_2},m_X=m_{\bm{Z_2}}=0)\times\\
    &{\rm P}(m_X=0|y,\bm{z_1},\bm{z_2},m_Y=m_{\bm{Z_2}}=0){\rm P}(m_{\bm{Z_2}}=0|x,y,\bm{z_1},m_X=m_Y=0).
\end{aligned}    
\end{equation*}
\end{proof}

\begin{corollary}\label{cor.1}
In m-DAG C, ${\rm P}(\bm{m}=\bm{0}|x,\bm{z_1},\bm{z_2})$ is recoverable.
\end{corollary}
\begin{proof}
By equation (\ref{A.eq:2}), the distribution ${\rm P}(\bm{m}=\bm{0}|x,\bm{z_1},\bm{z_2})$ can be rewritten as 
\begin{equation*}
    {\rm P}(\bm{m}=\bm{0}|x,\bm{z_1},\bm{z_2})=\frac{{\rm P}(\bm{m}=\bm{0},x,\bm{z_1},\bm{z_2})}{\sum_{y'}\frac{{\rm P}(\bm{m}=\bm{0},x,y',\bm{z_1},\bm{z_2})}{{\rm P}(\bm{m}=\bm{0}|x,y',\bm{z_1},\bm{z_2})}},
\end{equation*}
where every factor is recoverable by Lemma \ref{lemma.1}.
\end{proof}
\begin{corollary}\label{cor.2}
In m-DAG C, ${\rm P}(\bm{m}=\bm{0}|\bm{z_1},\bm{z_2})$ is recoverable.
\end{corollary}
\begin{proof}
By equation (\ref{A.eq:2}), the distribution ${\rm P}(\bm{m}=\bm{0}|\bm{z_1},\bm{z_2})$ can be rewritten as 
\begin{equation*}
    {\rm P}(\bm{m}=\bm{0}|\bm{z_1},\bm{z_2})=\frac{{\rm P}(\bm{m}=\bm{0},\bm{z_1},\bm{z_2})}{\sum_{x'}\frac{{\rm P}(\bm{m}=\bm{0},x',\bm{z_1},\bm{z_2})}{{\rm P}(\bm{m}=\bm{0}|x',\bm{z_1},\bm{z_2})}},
\end{equation*}
where every factor is recoverable by \ref{cor.1}.
\end{proof}
\begin{corollary}\label{cor.3}
In m-DAG C, ${\rm P}(x|\bm{z_1},\bm{z_2})$ is recoverable.
\end{corollary}
\begin{proof}
By equation (\ref{A.eq:1}), the distribution ${\rm P}(x|\bm{z_1},\bm{z_2})$ can be rewritten as 
\begin{equation*}
    {\rm P}(x|\bm{z_1},\bm{z_2})=\frac{{\rm P}(x|\bm{z_1},\bm{z_2},\bm{m}=\bm{0}){\rm P}(\bm{m}=\bm{0}|\bm{z_1},\bm{z_2})}{{\rm P}(\bm{m}=\bm{0}|x,\bm{z_1},\bm{z_2})},
\end{equation*}
where every factor is recoverable by Corollary \ref{cor.2}.
\end{proof}

\subsection*{A.6.\enspace Recoverability for ACE in m-DAG D}
The ACE is not recoverable in m-DAG D shown in Figure \ref{m-DAG.D}.

\begin{proof}
We prove the non-recoverability of the ACE in m-DAG D by first proving the non-recoverability of the ACE in simplified m-DAG D' shown in Figure \ref{m-DAG.D'}, where $Z_2$ is a binary confounder that causes its missingness. 

In m-DAG D', the marginal distribution of $Z_2$ is not recoverable by Theorem \ref{non-rec-V}, as $Z_2$ and its missingness indicator $M_{Z_2}$ are neighbours. Whereas, the conditional distribution of outcome given exposure and confounder is recoverable in m-DAG D', i.e. ${\rm P}(y|x,z_2)={\rm P}(y|x,z_2,\bm{m}=\bm{0})$. Using a derivation similar to that after expression (\ref{A.B.1}), the marginal distribution of the potential outcome can be expressed as: 
\begin{equation*}
{\rm P}(y|\hat{x})=\sum_z {\rm P}(y|x,z_2){\rm P}(z_2)=[{\rm P}(y|x,z_2=0)-{\rm P}(y|x,z_2=1)]{\rm P}(z_2=0)+{\rm P}(y|x,z_2=1),
\end{equation*}
where ${\rm P}(y|x,z_2=0)-{\rm P}(y|x,z_2=1)\neq 0$ given that $Y \not\!\perp\!\!\!\perp Z_2|X$ in m-DAG D'. Thus, the first term is non-recoverable given the second statement in Corollary \ref{non-rec} and ${\rm P}(y|\hat{x})$ is also non-recoverable given the first statement in Corollary \ref{non-rec}. Therefore, the ACE is non-recoverable in m-DAG D' and m-DAG D. The latter is given by Lemma \ref{G.G'}.
\end{proof}

\subsection*{A.7.\enspace Recoverability for ACE in m-DAG G}
The ACE is conjectured not recoverable in m-DAG G shown in Figure \ref{m-DAG.G}.

\begin{proof}
Expand the potential outcome as
\begin{equation*}\tag{\ref{A.B.1}}
    {\rm P}(y|\hat{x})=\sum_{\bm{z_1}}\sum_{\bm{z_2}}{\rm P}(y|\hat{x},\bm{z_1},\bm{z_2}){\rm P}(\bm{z_1},\bm{z_2}|\hat{x}).
\end{equation*}
Given that $Y\ind M_X|X,\bm{Z_1},\bm{Z_2},M_Y $ in $G_{\overline{X}}$ and $Y\ind M_{\bm{Z_2}}|X,\bm{Z_1},\bm{Z_2},M_Y,M_X$ in $G_{\overline{X}}$, by Rule-1
\begin{equation*}
    {\rm P}(y|\hat{x},\bm{z_1},\bm{z_2})={\rm P}(y|\hat{x},\bm{z_1},\bm{z_2},m_X=0,m_{\bm{Z_2}}=0).
\end{equation*}
However, this probability is not recoverable by Corollary \ref{non-rec-C}, as $Y$ and its missingness indicator $M_{Y}$ are neighbours. The second term in (\ref{A.B.1}) is recoverable (proof as for m-DAG B), so from first statement in Corollary \ref{non-rec}, each term in the sum is non-recoverable, that is, (\ref{A.B.1}) is a sum of non-recoverable terms. Therefore we conjecture the ACE is not recoverable in m-DAG G.
\end{proof}

\subsection*{A.8.\enspace Recoverability for ACE in other m-DAGs}
The ACE is recoverable in m-DAG A and the marginal distribution of potential outcome is given by ${\rm P}(y|\hat{x})=\sum_{\bm{z_1}}\sum_{\bm{z_2}}{\rm P}(y|x,\bm{z_1},\bm{z_2},\bm{m}=\bm{0}){\rm P}(\bm{z_1}){\rm P}(\bm{z_2}|\bm{z_1},m_{\bm{Z_2}}=0)$. The ACE is not recoverable in m-DAGs E, F, I and J given the non-recoverability in m-DAG D based on Lemma \ref{G.G'}, since they result from adding edges to m-DAG D. If our conjecture on the non-recoverability of the ACE in m-DAG G stands, we can conclude that the ACE is not recoverable in m-DAGs H and I for the same reason. 

\subsection*{A.9.\enspace Recoverability for ACE in expanded version of m-DAG D}
The ACE is recoverable in m-DAG D'' shown in Figure \ref{m-DAG.D''}.

\begin{proof}
The marginal distribution of the potential outcome ${\rm P}(y|\hat{x})$ is given by
\begin{equation}\label{A.D.1}
    {\rm P}(y|\hat{x})=\sum_{\bm{z_1}}\sum_{\bm{z_2}}\sum_{\bm{z_3}}{\rm P}(y|\hat{x},\bm{z_1},\bm{z_2},\bm{z_3}){\rm P}(\bm{z_1},\bm{z_2},\bm{z_3}|\hat{x}).
\end{equation}
Following Rule-1, we can condition the first term condition on all missingness indicators. We use $\bm{m}=\bm{0}$ as a shorthand of $\{m_X=m_Y=m_{\bm{Z_2}}=m_{\bm{Z_3}}=0\}$ in this section.
\begin{equation*}
\begin{aligned}
    {\rm P}(y|\hat{x},\bm{z_1},\bm{z_2},\bm{z_3})
    &={\rm P}(y|\hat{x},\bm{z_1},\bm{z_2},\bm{z_3},m_X=0,m_Y=0,m_{\bm{Z_2}}=0,m_{\bm{Z_3}}=0)\\
    &={\rm P}(y|\hat{x},\bm{z_1},\bm{z_2},\bm{z_3},\bm{m}=\bm{0}),
\end{aligned}
\end{equation*}
as $Y\ind M_Y|X,\bm{Z_1},\bm{Z_2},\bm{Z_3}$ in $G_{\overline{X}}$, $Y\ind M_X|X,\bm{Z_1},\bm{Z_2},\bm{Z_3},M_Y$ in $G_{\overline{X}}$, and $Y\ind M_{\bm{Z_2}}|X,\bm{Z_1},\bm{Z_2},\bm{Z_3},M_Y,M_X$ in $G_{\overline{X}}$. Applying Rule-2, 
\begin{equation*}
    {\rm P}(y|\hat{x},\bm{z_1},\bm{z_2},\bm{z_3},\bm{m}=\bm{0})={\rm P}(y|x,\bm{z_1},\bm{z_2},\bm{z_3},\bm{m}=\bm{0}),
\end{equation*}
since $Y\ind X|\bm{Z_1},\bm{Z_2},\bm{Z_3},M_Y,M_X,M_{\bm{Z_2}},M_{\bm{Z_3}}$ in $G_{\underline{X}}$. Hence, the first term of (\ref{A.D.1}) is recoverable.

The second term of (\ref{A.D.1}) can be simplified as
\begin{equation*}
    {\rm P}(\bm{z_1},\bm{z_2},\bm{z_3}|\hat{x})={\rm P}(\bm{z_1},\bm{z_2},\bm{z_3})={\rm P}(\bm{z_1},\bm{z_3}){\rm P}(\bm{z_2}|\bm{z_1},\bm{z_3}),
\end{equation*}
by Rule-3, given $(\bm{Z_1},\bm{Z_2},\bm{Z_3})\ind X$ in $G_{\overline{X}}$. Given $\bm{Z_2}\ind M_{\bm{Z_2}}|\bm{Z_1},\bm{Z_3}$ in $G$, 
\begin{equation*}
{\rm P}(\bm{z_2}|\bm{z_1},\bm{z_3})={\rm P}(\bm{z_2}|\bm{z_1},\bm{z_3},m_{\bm{Z_2}}=0).
\end{equation*}
Then we can express ${\rm P}(\bm{z_2}|\bm{z_1},\bm{z_3})$ with the use of equation (\ref{A.eq:1}) as
\begin{equation*}
    {\rm P}(\bm{z_2}|\bm{z_1},\bm{z_3},m_{\bm{Z_2}}=0)=\frac{{\rm P}(\bm{z_2}|\bm{z_1},\bm{z_3},m_{\bm{Z_2}}=m_{\bm{Z_3}}=0){\rm P}(m_{\bm{Z_3}}=0|\bm{z_1},\bm{z_3},m_{\bm{Z_2}}=0)}{{\rm P}(m_{\bm{Z_3}}=0|\bm{z_1},\bm{z_2},\bm{z_3},m_{\bm{Z_2}}=0)},
\end{equation*}
where ${\rm P}(\bm{z_2}|\bm{z_1},\bm{z_3},m_{\bm{Z_2}}=m_{\bm{Z_3}}=0)$ is recoverable. Given $M_{\bm{Z_3}} \ind M_{\bm{Z_2}}|\bm{Z_1},\bm{Z_3}$ in $G$, the second term of numerators can be expressed as
\begin{equation*}
{\rm P}(m_{\bm{Z_3}}=0|\bm{z_1},\bm{z_3},m_{\bm{Z_2}}=0)={\rm P}(m_{\bm{Z_3}}=0|\bm{z_1},\bm{z_3})=\frac{{\rm P}(m_{\bm{Z_3}}=0,\bm{z_1},\bm{z_3})}{{\rm P}(\bm{z_1},\bm{z_3})}.
\end{equation*}
Given $M_{\bm{Z_3}} \ind \bm{Z_3}|\bm{Z_1},\bm{Z_2},M_{\bm{Z_2}}$ in $G$, rewrite the denominator as
\begin{equation*}
{\rm P}(m_{\bm{Z_3}}=0|\bm{z_1},\bm{z_2},\bm{z_3},m_{\bm{Z_2}}=0)={\rm P}(m_{\bm{Z_3}}=0|\bm{z_1},\bm{z_2},m_{\bm{Z_2}}=0).
\end{equation*}
Therefore, the second term of (\ref{A.D.1}) can be simplified as
\begin{equation*}
    {\rm P}(\bm{z_1},\bm{z_3}){\rm P}(\bm{z_2}|\bm{z_1},\bm{z_3})=\frac{{\rm P}(\bm{z_2}|\bm{z_1},\bm{z_3},m_{\bm{Z_2}}=m_{\bm{Z_3}}=0){\rm P}(m_{\bm{Z_3}}=0,\bm{z_1},\bm{z_3})}{{\rm P}(m_{\bm{Z_3}}=0|\bm{z_1},\bm{z_2},m_{\bm{Z_2}}=0)},
\end{equation*}
where each term is recoverable. Hence, the marginal distribution of the potential outcome ${\rm P}(y|\hat{x})$ is recoverable in m-DAG D'', by the following expression
\begin{equation*}
{\rm P}(y|\hat{x})=\sum_{\bm{z_1}}\sum_{\bm{z_2}}\sum_{\bm{z_3}}{\rm P}(y|x,\bm{z_1},\bm{z_2},\bm{z_3},\bm{m}=\bm{0})\frac{{\rm P}(\bm{z_2}|\bm{z_1},\bm{z_3},m_{\bm{Z_2}}=m_{\bm{Z_3}}=0){\rm P}(m_{\bm{Z_3}}=0,\bm{z_1},\bm{z_3})}{{\rm P}(m_{\bm{Z_3}}=0|\bm{z_1},\bm{z_2},m_{\bm{Z_2}}=0)}.
\end{equation*}
\end{proof}

\begin{sidewaystable}[h!]
\caption{Descriptive statistics for the analysis variables for the case study, using data from Victorian Adolescent Health Cohort Study ($n=961$)}
\label{tab1}
\centering
\begin{tabular}{cclccc}
\hline
\multicolumn{3}{c}{}                                        & \multicolumn{2}{c}{Stratified by   exposure\textsuperscript{a}} &              \\
Role       & Label & Variable                               & Unexposed              & Exposed               & Missing (\%) \\ 
\hline
Exposure   & $X$   & Cannabis use in adolescence, Yes                      & 603 (62.7)             & 84 (8.7)              & 28.5            \\
Outcome    & $Y$   & Adulthood mental health score \textsuperscript{b}   & -0.11 (1.00)           & 0.48 (0.87)           & 10.4         \\
Confounder & $C_1$  & Parental education (failure to complete high-school), Yes & 206 (34.2)             & 35 (41.7)             & 0            \\
Confounder & $C_2$  & Parental divorce or separation, Yes                  & 94 (15.6)             & 38 (45.2)             & 0            \\
Confounder & $C_3$   & Antisocial behaviour in adolescence, Yes              & 42 (7.0)               & 31 (36.9)             & 0            \\
Confounder & $C_4$  & Adolescent depression \& anxiety, Yes   & 297 (49.3)             & 62 (81.6)             & 12.4         \\
Confounder & $C_5$   & Alcohol use in adolescence, Yes                       & 155 (25.7)             & 65 (87.8)             & 19.3         \\
Auxiliary  & $A$   & Participant’s age at wave two \textsuperscript{b}       & -0.10 (0.81)           & 0.14 (0.94)           & 8.9          \\ 
\hline
\end{tabular}

\raggedright
a. For incomplete variables, the descriptive statistics are obtained from the records with available data on the given variable.

b. In standard deviation units, standardised to the overall sample.
\end{sidewaystable}

\begin{sidewaystable}[h!]
\caption{The recoverability results for the ACE and the expression for the marginal distribution of the potential outcomes in terms of observable data}
\label{Recoverability}
\centering
\begin{tabular}{ccc}
\hline
m-DAG & Expression of potential outcome\textsuperscript{a, b, c, d} & Recoverability of ACE \\ \hline
A & $\sum_{\bm{z_1},\bm{z_2}}\rm{P}(y|x,\bm{z_1},\bm{z_2},\bm{m}=\bm{0})\rm{P}(\bm{z_1})\rm{P}(\bm{z_2}|\bm{z_1},m_{\bm{Z_2}}=0)$ & Yes\\
B &  $\sum_{\bm{z_1},\bm{z_2}}\rm{P}(y|x,\bm{z_1},\bm{z_2},\bm{m}=\bm{0}) \sum_{x'}\frac{\rm{P}(m_X=0,m_{\bm{Z_2}}=0,x',\bm{z_1},\bm{z_2})}{\rm{P}(m_{\bm{Z_2}}=0|x',\bm{z_1},m_X=0)\rm{P}(m_X=0|\bm{z_1},\bm{z_2},m_{\bm{Z_2}}=0)}$ & Yes \\
C &  $\sum_{\bm{z_1},\bm{z_2}}\frac{\rm{P}(y|x,\bm{z_1},\bm{z_2},\bm{m}=\bm{0})\rm{P}(x,\bm{z_1},\bm{z_2},m_X=m_{\bm{Z_2}}=0)}{\rm{P}(m_{\bm{Z_2}}=0|x,y,\bm{z_1},m_X=m_Y=0)\rm{P}(m_X=0|y,\bm{z_1},\bm{z_2},m_Y=m_{\bm{Z_2}}=0)} \frac{\sum_{x',y''}\frac{\rm{P}(\bm{m}=\bm{0},x,y'',\bm{z_1},\bm{z_2})} {\rm{P}(\bm{m}=\bm{0}|x',y'',\bm{z_1},\bm{z_2})}} {\sum_{y'}\frac{\rm{P}(\bm{m}=\bm{0},x,y',\bm{z_1},\bm{z_2})} {\rm{P}(\bm{m}=\bm{0}|x,y',\bm{z_1},\bm{z_2})}}$ & Yes \\
D &  & No \\
E &  & No \\
F &  & No \\
G &  & Conjecture no \\
H &  & Conjecture no \\
I &  & No \\
J &  & No \\ \hline
\end{tabular}

\raggedright
Abbreviation: m-DAG, missingness directed acyclic graph; ACE, average causal effect.

a. Proofs are provided in Supplemental Material.

b. $\bm{m}=\bm{0}$ is a shorthand of $\{ m_X=m_Y=m_{\bm{Z_2}}=0 \}$. 

c. $\rm{P}(\bm{m}=\bm{0}|x,y,\bm{z_1},\bm{z_2})=\rm{P}(m_Y=0|x,\bm{z_1},\bm{z_2},m_X=m_{\bm{Z_2}}=0)\rm{P}(m_X=0|y,\bm{z_1},\bm{z_2},m_Y=m_{\bm{Z_2}}=0)\rm{P}(m_{\bm{Z_2}}=0|x,y,\bm{z_1},m_X=m_Y=0)$.

d. A blank space is left where the distribution is not recoverable or it has not been established.
\end{sidewaystable}

\begin{table}[p]
\caption{The relative strength of exposure-confounder interactions ($\beta_7$ and $\beta_8$) compared with the exposure effect ($\beta_6$) in six outcome scenarios}
\label{outscen}
\centering
\begin{tabular}{ccccccc}
\hline
\multirow{2}{*}{The relative strength}     & \multicolumn{6}{c}{Outcome scenarios} \\
  & I   & II   & III   & IV    & V  & VI  \\ \hline
\begin{tabular}[c]{@{}c@{}}For complete confounder $C_3$,\\ the ratio of $\beta_7/\beta_6$\end{tabular}   & 0   & 0.5  & 0     & 0.5   & 3  & 0   \\
\begin{tabular}[c]{@{}c@{}}For incomplete confounder $C_4$,\\ the ratio of $\beta_8/\beta_6$\end{tabular} & 0   & 0    & -0.5  & -0.5  & 0  & -3  \\ \hline
\end{tabular}

\raggedright
a. The signs of the exposure-confounder interactions were kept as for the estimated regression coefficients in (\ref{outgen}) from the VAHCS data. 
\end{table}

\begin{table}[p]
\caption{Description and specification of missing data methods under comparison}
\label{methods}
\centering
\begin{tabular}{cm{0.8\linewidth}}
\hline
Label  & \multicolumn{1}{c}{Methods}                                                                                                                                                                                                                                             \\ \hline
CCA    & Complete-case analysis                                                                                                                                                                                                                                                  \\ \hline
MI-Sim & Simple MI: Imputation model uses all variables without any interactions.                                                                                                                                                                                                \\ \hline
MI-EO  & Exposure-outcome interaction MI: Imputation model for incomplete confounders added $XY$ interaction.                                                                                                                                                                    \\ \hline
MI-EI  & Exposure-incomplete confounder interaction MI: Imputation model for $C_4$ added $XC_5$ interaction, and for $C_5$ added $XC_4$ interaction. Imputation model for the outcome added both $XC_4$ and $XC_5$ interactions.                                                 \\ \hline
MI-EC  & Exposure-confounder interaction MI: Imputation model for incomplete confounders added interactions between exposure and other confounders. Imputation model for the outcome added all exposure-confounder interactions.                                                           \\ \hline
MI-Com & Approximately compatible MI: Imputation model and outcome model in g-computation were made to be approximately compatible by including in univariate imputation models all relevant two-way confounder-confounder, outcome-confounder and exposure-outcome interactions. \\ \hline
MI-SMC & Substantive model compatible FCS MI: Imputation model is compatible with the outcome model in g-computation.                                                                                                                 \\ \hline
\end{tabular}

\raggedright
a. All MI approaches were carried out by fully conditional specification  (FCS) approach.
\end{table}

\begin{table}[p]
\caption{Estimates of ACE obtained using missing data methods in the case study ($n=961$)}
\label{case}
    \centering
    \begin{tabular}{c|ccc}
\hline
Method & Estimate & Standard error & 95\% confidence interval \\ \hline
CCA    & 0.18 & 0.17 & -0.16, 0.50 \\
MI-Sim & 0.25 & 0.15 & -0.05, 0.54 \\  
MI-EO  & 0.22 & 0.15 & -0.08, 0.51 \\
MI-EI  & 0.25 & 0.15 & -0.04, 0.55 \\   
MI-EC  & 0.25 & 0.15 & -0.04, 0.53 \\   
MI-Com & 0.19 & 0.15 & -0.10, 0.49 \\  
MI-SMC & 0.22 & 0.15 & -0.07, 0.50 \\
\hline    
\end{tabular}
\end{table}

\begin{sidewaystable}[h!]
\caption{Assessment of the existence of an arrow from each incomplete variable to each missingness indicator in the example from Victorian Adolescent Health Cohort Study ($n=961$)}
\label{VAHCS.assumption}
\centering
\begin{tabular}{cc|ccc}
\hline
\multicolumn{2}{c|}{\multirow{2}{*}{}}                                                                                                                                         & \multicolumn{3}{c}{Arrows to:}                                                                                                                                                                                                                                                                                                                                                                                                                                                                                                                                                                                         \\ \cline{3-5} 
\multicolumn{2}{c|}{}                                                                                                                                                          & \multicolumn{1}{c|}{$M_{C_3},M_{C_4}$}                                                                                                                                                                                                      & \multicolumn{1}{c|}{$M_X$}                                                                                                                                                                                                           & $M_Y$                                                                                                                                                \\ \hline
\multicolumn{1}{c|}{\multirow{8}{*}{\parbox[t]{2mm}{\multirow{3}{*}{\rotatebox[origin=c]{90}{Arrow from:}}}}} & \begin{tabular}[c]{@{}c@{}} $C_3,C_4$ \\ \\ (Adolescent depression \\ \& anxiety and Alcohol \\ use at Wave 2-6) \\ \\ \end{tabular} & \multicolumn{1}{c|}{\begin{tabular}[c]{@{}c@{}}Likely\\  \\ Failure to return form can be due \\ to mental health issues\textsuperscript{a}. Failure \\ to answer alcohol use can be due \\ to stigma attached to it\textsuperscript{b-e}.\end{tabular}} & \multicolumn{1}{c|}{\begin{tabular}[c]{@{}c@{}}Likely\\  \\ Failure to return form can be due \\ to mental health issues or alcohol use\textsuperscript{a-e}. \\ \\ \\ \end{tabular}}                                                                   & \begin{tabular}[c]{@{}c@{}}Likely\\  \\ Failure to return form and attrition \\ by wave 7 can be due to mental \\ health issues\textsuperscript{a}. \\ \\ \end{tabular} \\ \cline{2-5} 
\multicolumn{1}{c|}{}                             & \begin{tabular}[c]{@{}c@{}} $X$ \\ \\ (Cannabis use at \\ Wave 2-6) \\ \\ \\ \end{tabular}                                               & \multicolumn{1}{c|}{\begin{tabular}[c]{@{}c@{}}Likely\\ \\ Failure to return form can be due \\ to cannabis use\textsuperscript{a,e}. \\ \\ \\ \end{tabular}}                                                                                          & \multicolumn{1}{c|}{\begin{tabular}[c]{@{}c@{}}Likely\\  \\ Failure to answer cannabis use can \\ be due to illegality and stigma \\ attached to it. Failure to return form \\ can be due to cannabis use\textsuperscript{a,e}. \end{tabular}} & \begin{tabular}[c]{@{}c@{}}Likely\\  \\ Failure to return form and attrition \\ by wave 7 can be due to cannabis \\ use\textsuperscript{d,e}. \\ \\ \end{tabular}       \\ \cline{2-5} 
\multicolumn{1}{c|}{}                             & \begin{tabular}[c]{@{}c@{}}$Y$\\ \\ (Adulthood mental \\ health score at Wave 7)\end{tabular}                             & \multicolumn{1}{c|}{\begin{tabular}[c]{@{}c@{}}Not likely\\  \\ Missingness in confounders preceded \\ outcome by at least six months.\end{tabular}}                                                                       & \multicolumn{1}{c|}{\begin{tabular}[c]{@{}c@{}}Not likely\\  \\ Missingness in exposure preceded \\ outcome by at least six months.\end{tabular}}                                                                              & \begin{tabular}[c]{@{}c@{}}Likely\\  \\ Failure to return form can be due to \\ mental health issues\textsuperscript{a}.\end{tabular}                         \\ \hline
\end{tabular}
\raggedright
a. \citep{cheung2017impact}\\
b. \citep{gorman2014assessing}\\
c. \citep{lemmens1988bias}\\
d. \citep{caetano2001non}\\
e. \citep{zhao2009non}
\end{sidewaystable}

\begin{figure}[h!]
\centering
\includegraphics[width=0.9\linewidth]{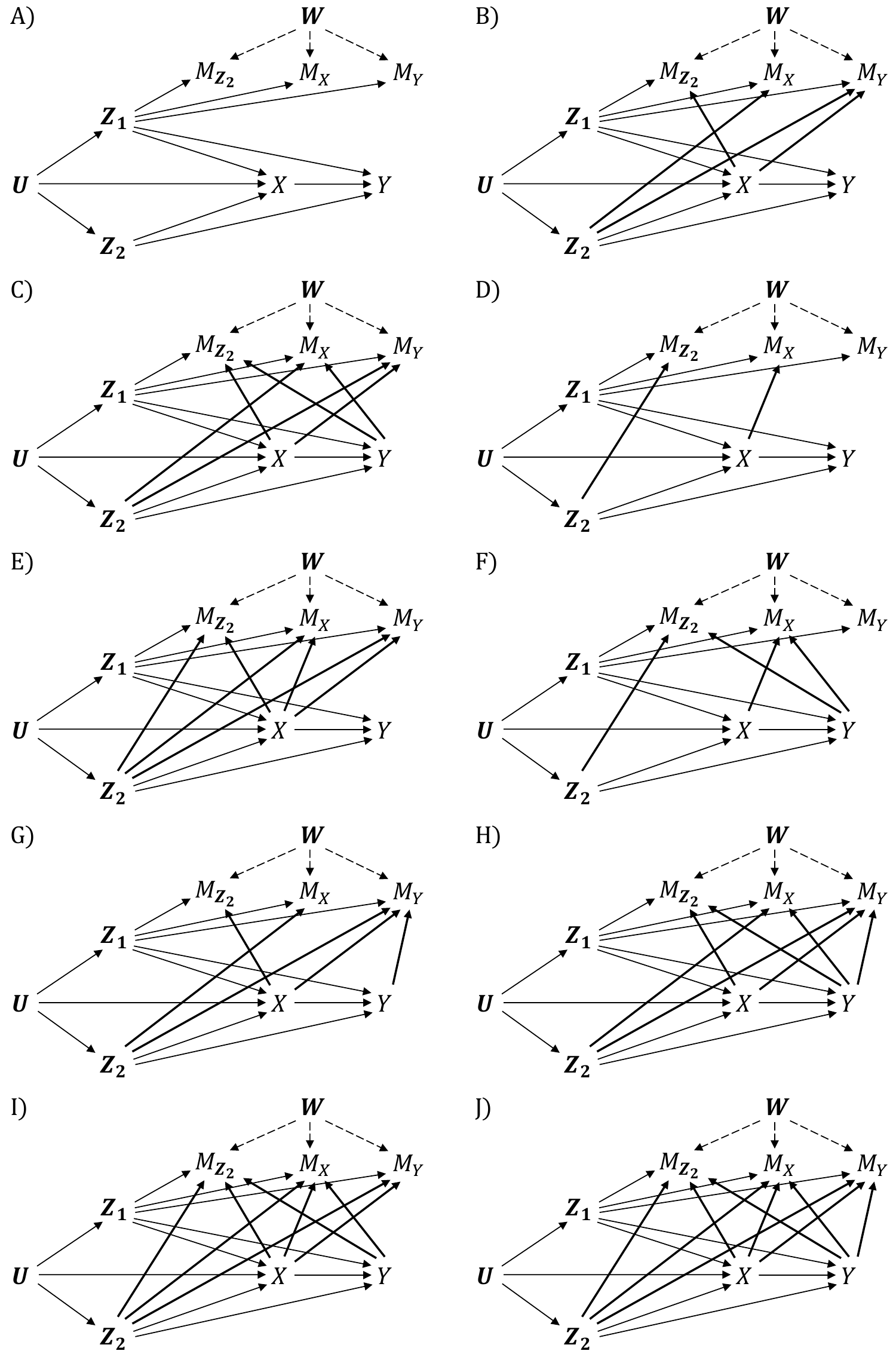}
\caption{Canonical missingness directed acyclic graphs (m-DAGs) for illustrating typical missingness mechanisms in point-exposure setting for epidemiological studies, adapted from \citep{moreno2018canonical}. The findings in Section \ref{the.res.section} of the manuscript assume the absence of the dashed arrows, whereas the findings of Section \ref{simu.res.section} assume these dashed arrows are present.}
\label{m-DAG}
\end{figure}

\begin{figure}[h!]
\centering
\includegraphics[width=\linewidth]{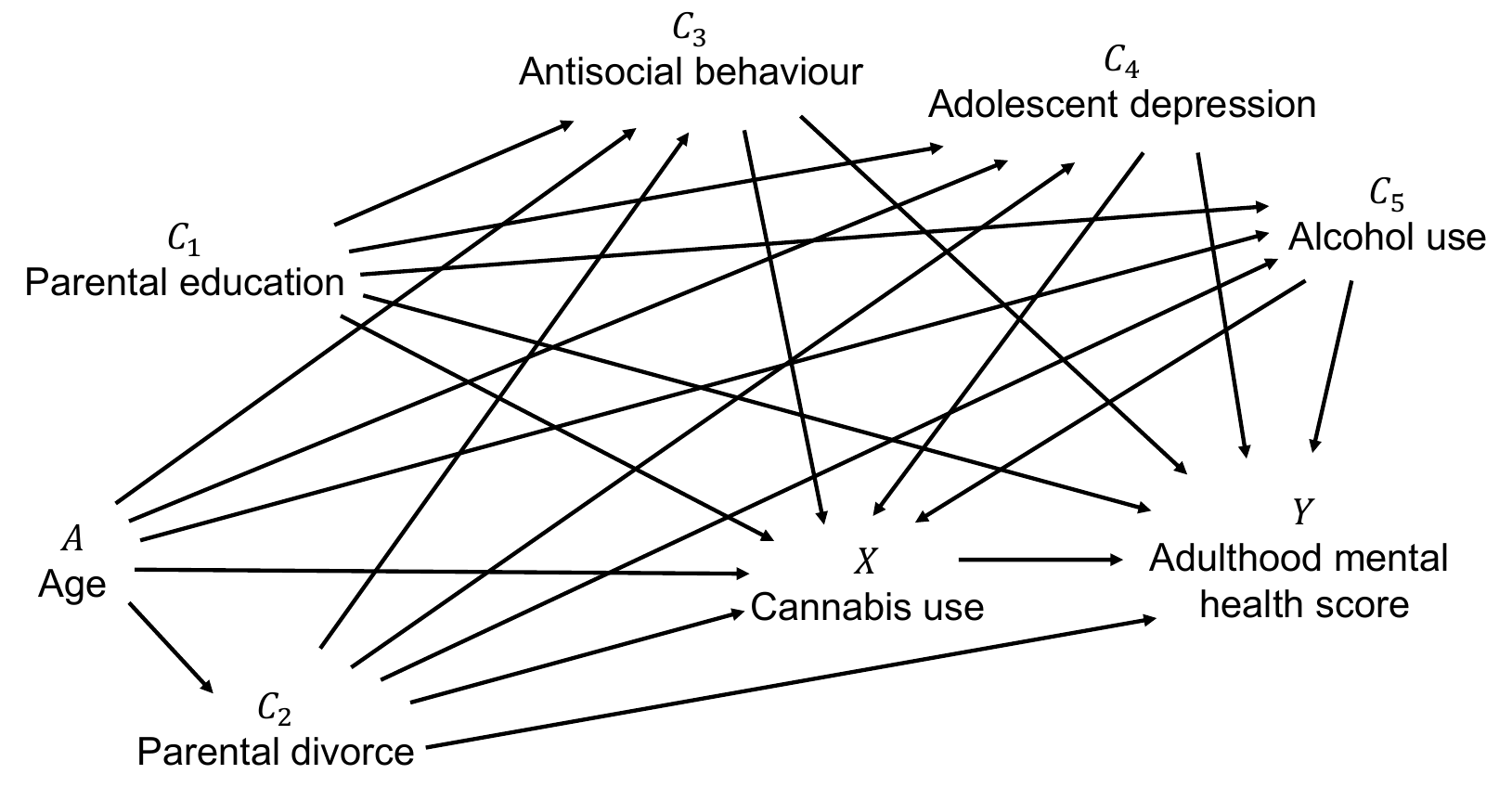}
\caption{Causal diagram guiding the complete data generation diagram in the simulation study.}
\label{simuDAG}
\end{figure}

\begin{figure}[h!]
\centerfloat
\includegraphics[width=\linewidth]{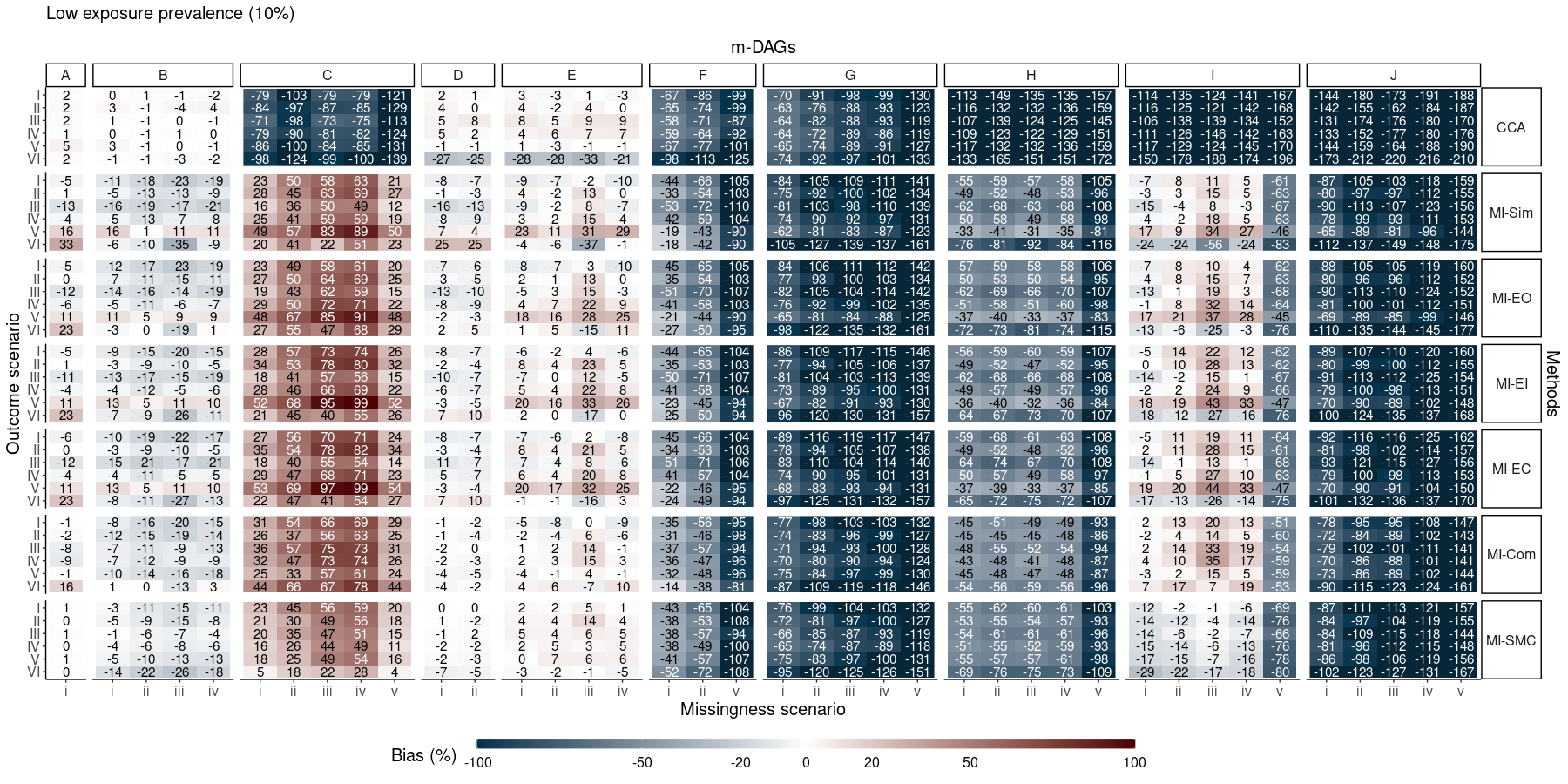}
\includegraphics[width=\linewidth]{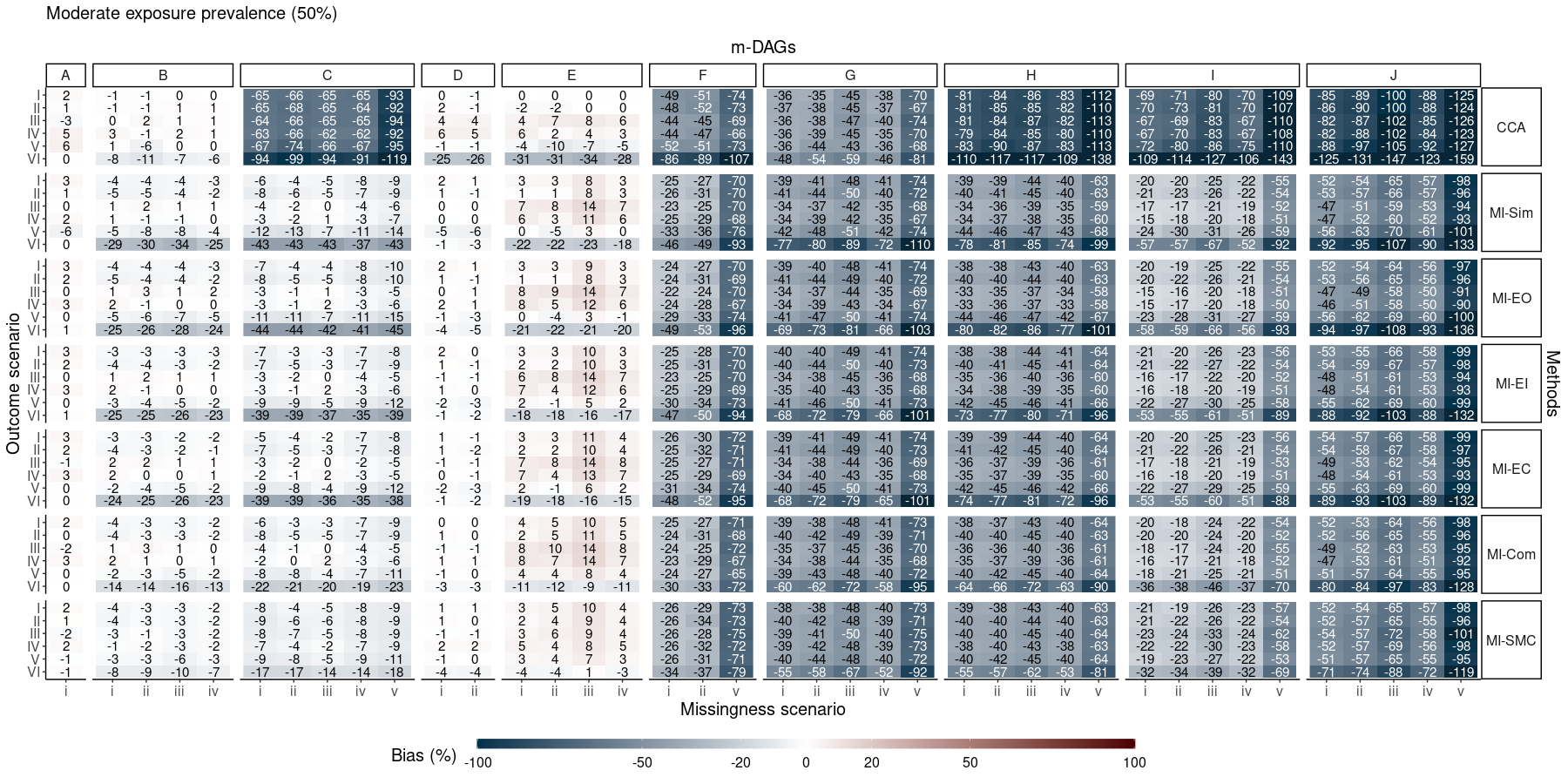}
\caption{Simulation study results: relative bias (\%) for missing data methods in canonical missingness directed acyclic graphs (m-DAGs) for low (top panel) and moderate (bottom panel) exposure prevalence scenarios, across a range of outcome and missingness scenarios.}
\label{RB}
\end{figure}

\begin{figure}[h!]
\centerfloat
\includegraphics[width=\linewidth]{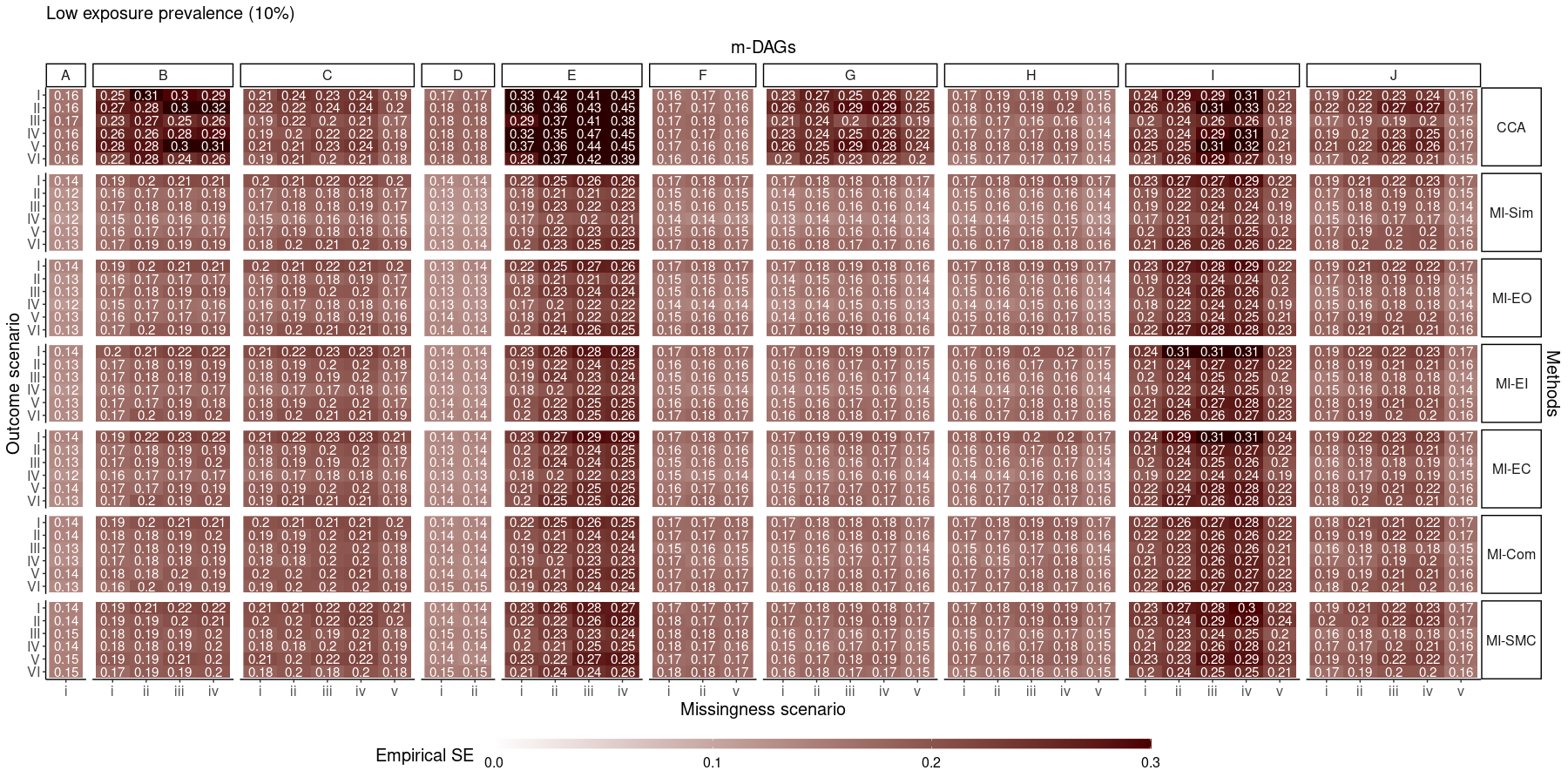}
\includegraphics[width=\linewidth]{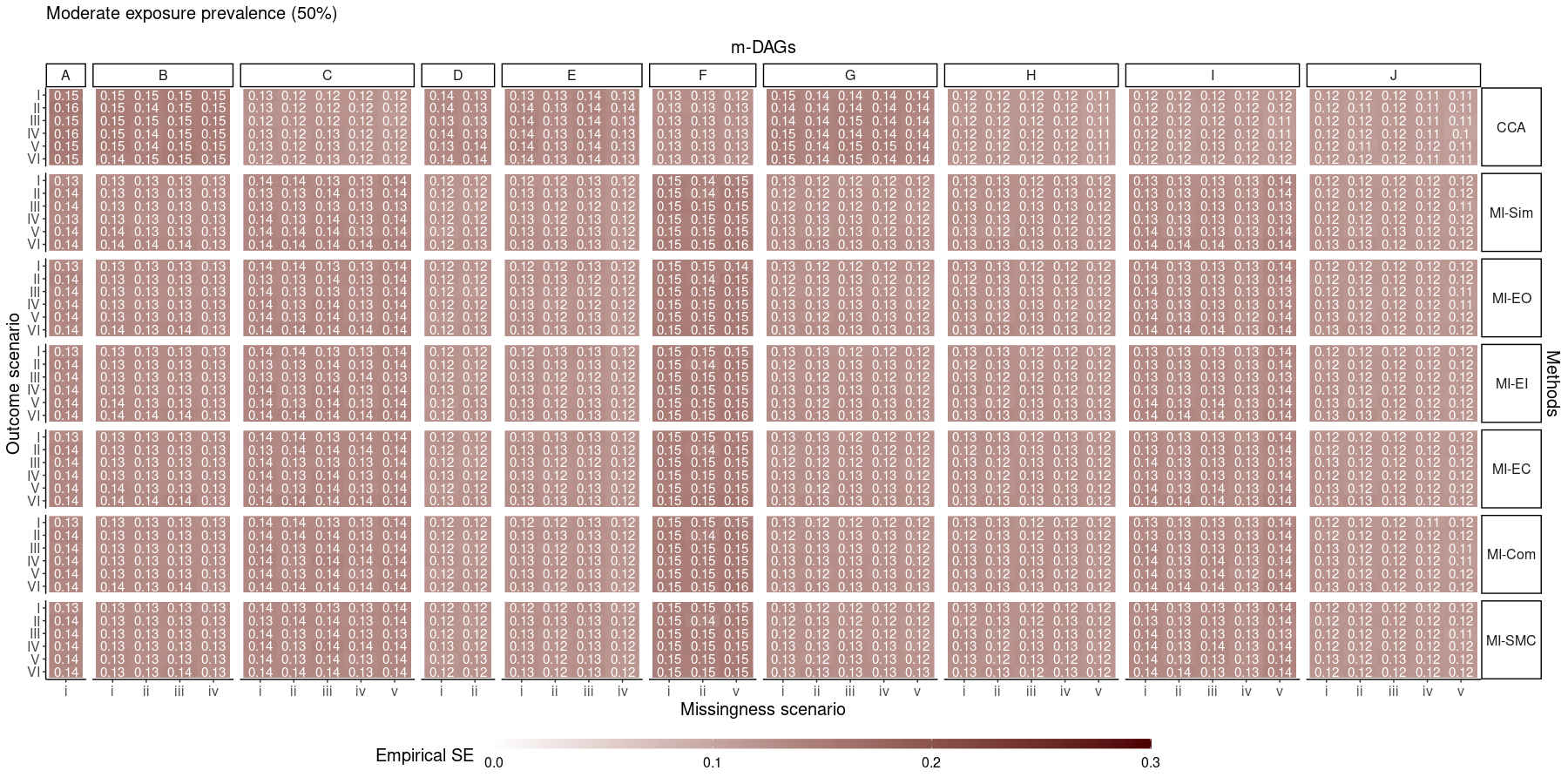}
\caption{Simulation study results: empirical standard error (SE) for missing data methods in canonical missingness directed acyclic graphs (m-DAGs) for low (top panel) and moderate (bottom panel) exposure prevalence scenarios, across a range of outcome and missingness scenarios.}
\label{EmpSE}
\end{figure}

\begin{figure}[h!]
\centerfloat
\includegraphics[width=\linewidth]{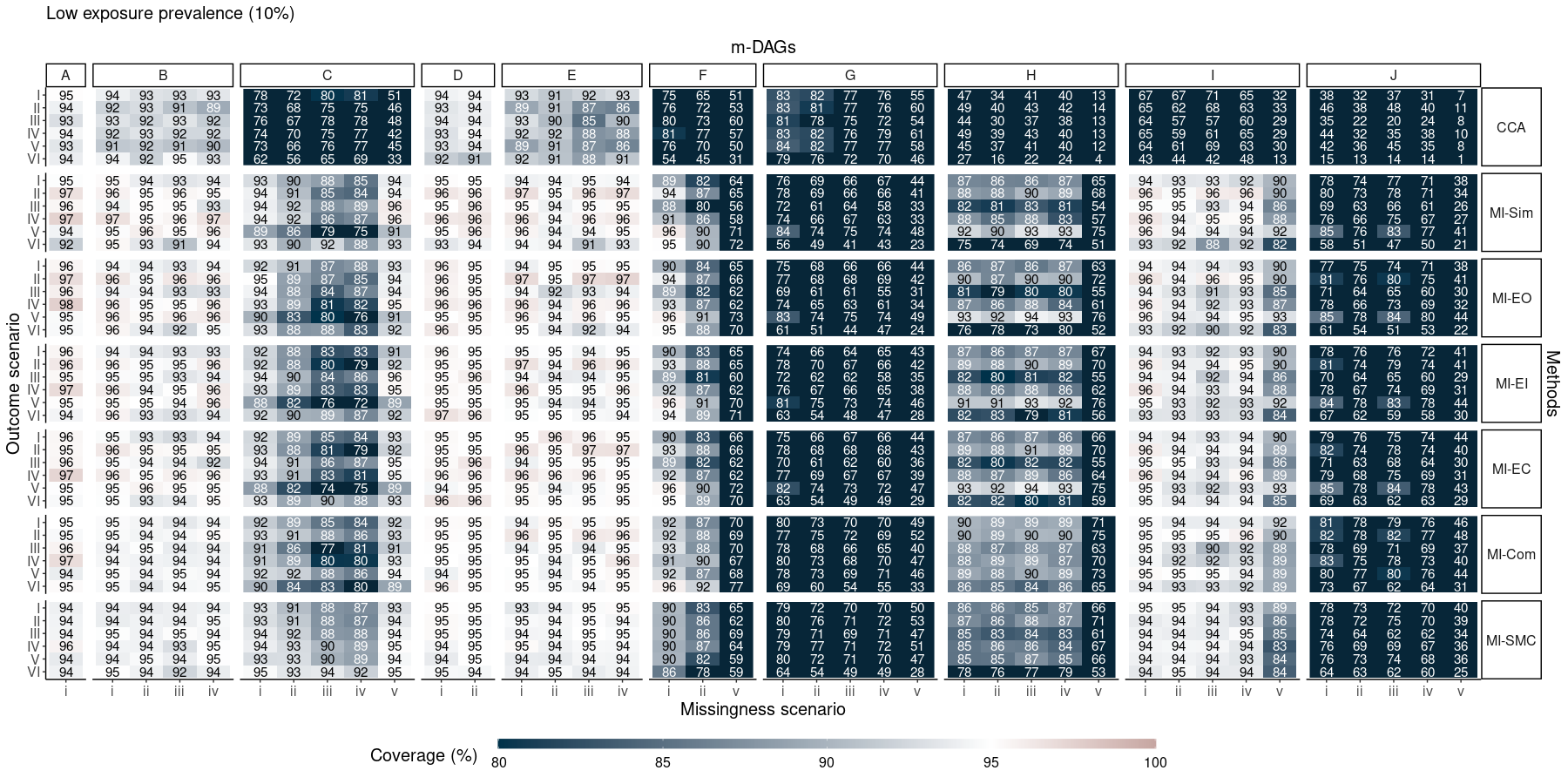}
\includegraphics[width=\linewidth]{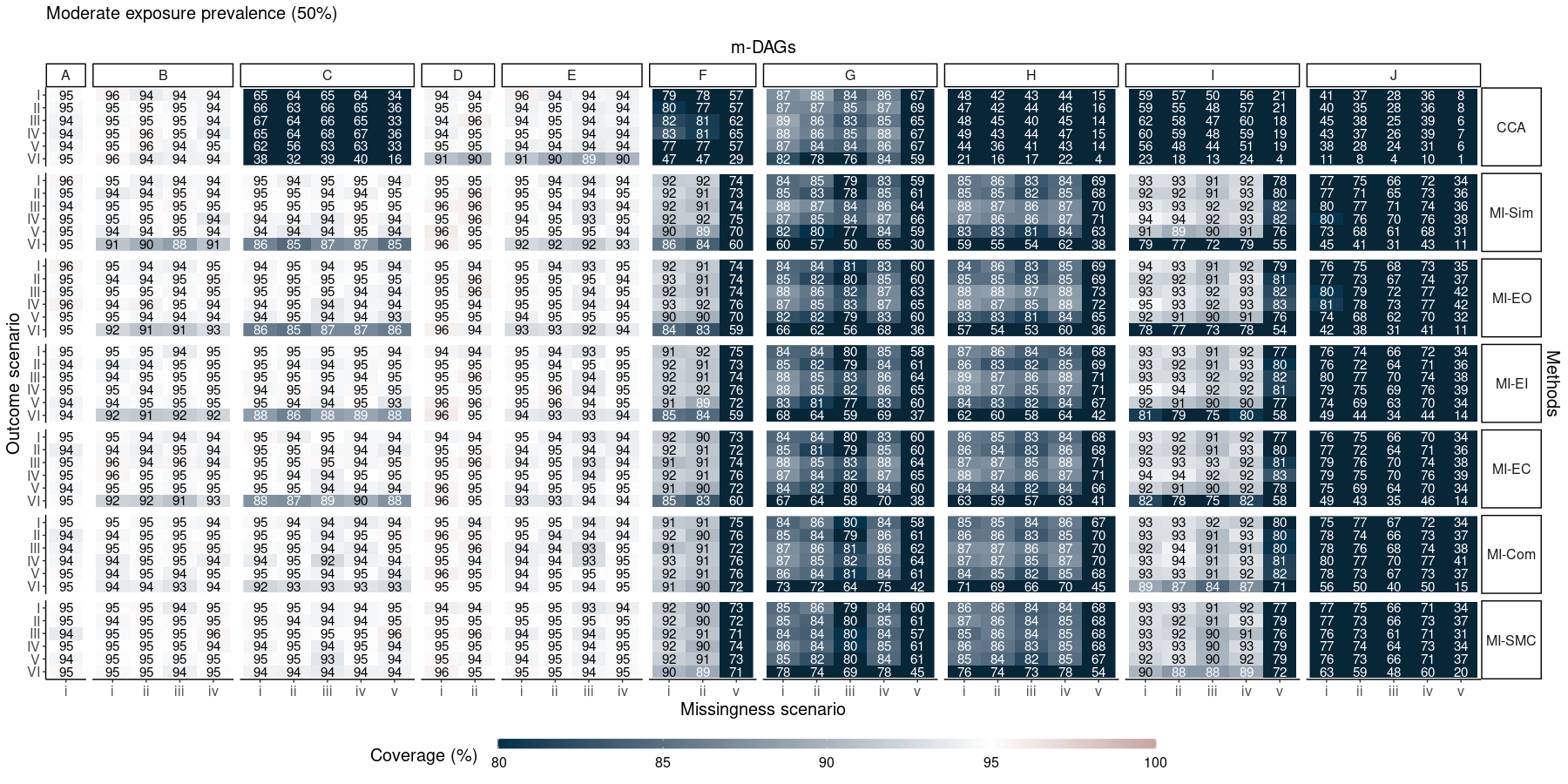}
\caption{Simulation study results: coverage probability (\%) for missing data methods in canonical missingness directed acyclic graphs (m-DAGs) for low (top panel) and moderate (bottom panel) exposure prevalence scenarios, across a range of outcome and missingness scenarios.}
\label{Coverage}
\end{figure}

\begin{figure*}[!htbp]
  \centering
  \includegraphics[width=0.5\linewidth]{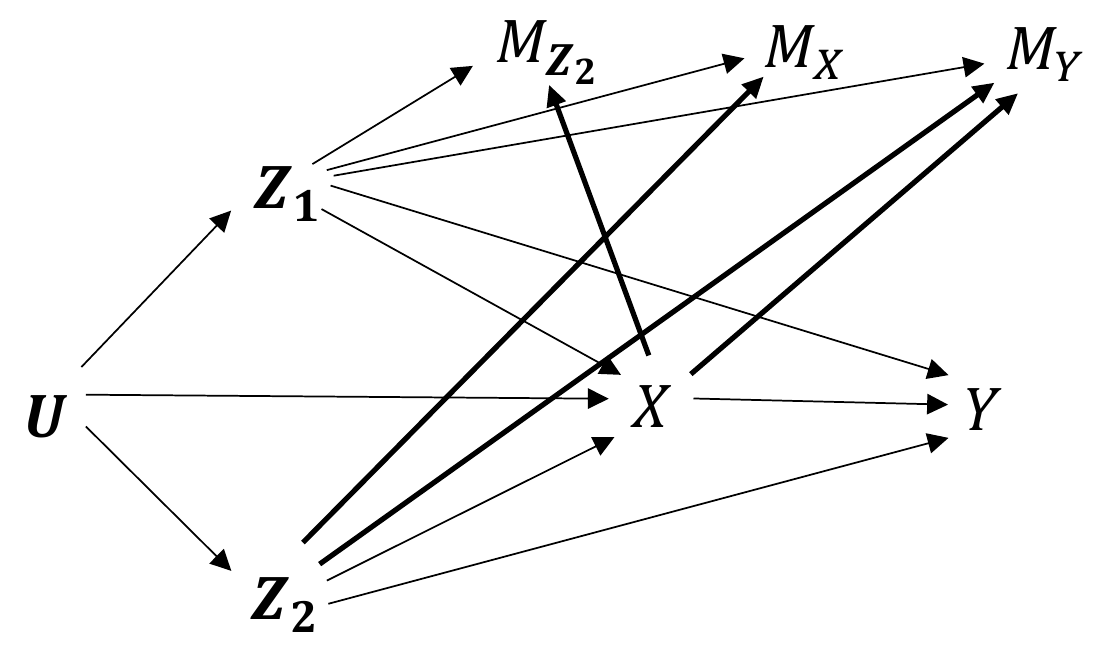}
  \caption{Canonical missingness directed acyclic graph (m-DAG) B for a general point-exposure study assuming no unmeasured common causes for missingness indicators}
  \label{m-DAG.B}
\end{figure*}

\begin{figure*}[!htbp]
  \centering
  \includegraphics[width=0.5\linewidth]{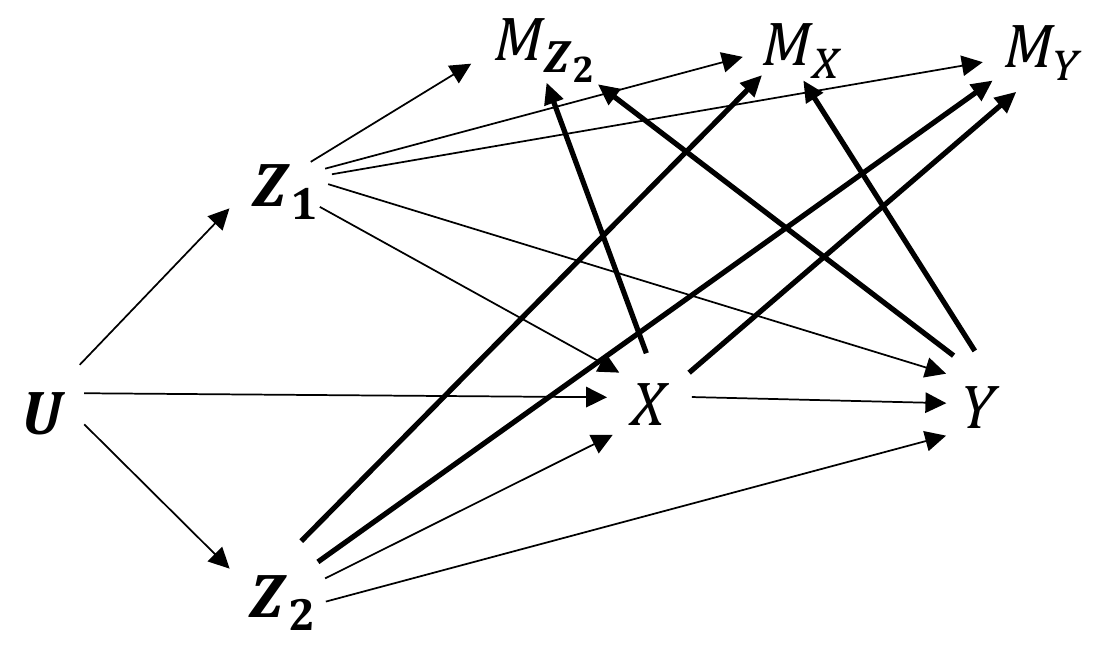}
  \caption{Canonical m-DAG C for a general point-exposure study assuming no unmeasured common causes for missingness indicators}
  \label{m-DAG.C}
\end{figure*}

\begin{figure*}[!htbp]
  \centering
  \includegraphics[width=0.5\linewidth]{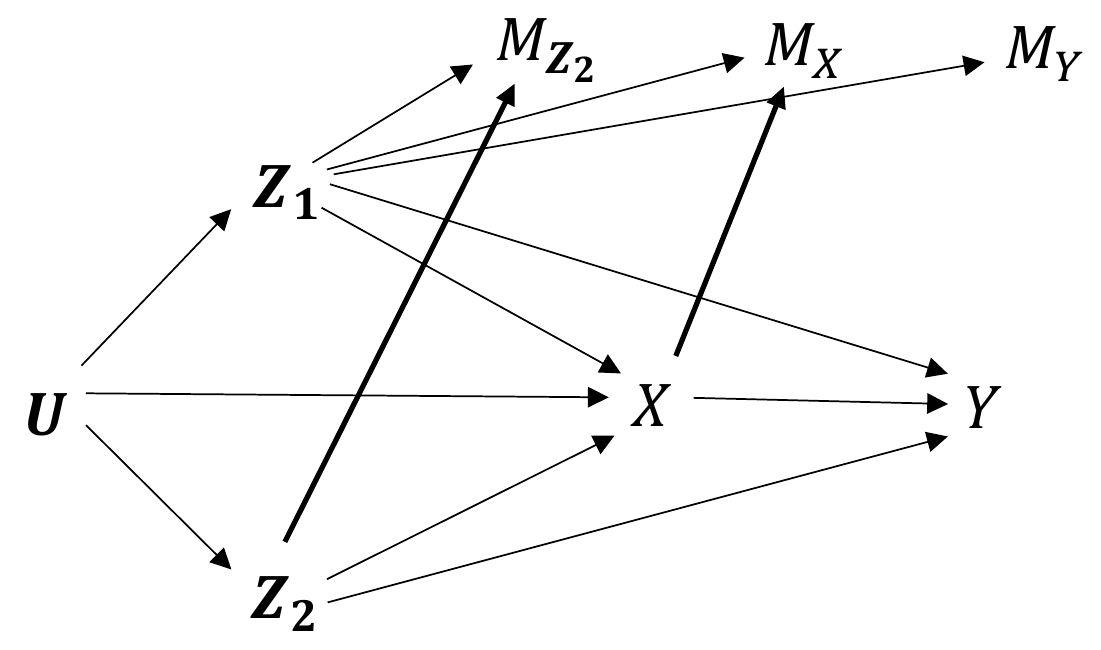}
  \caption{Canonical m-DAG D for a general point-exposure study assuming no unmeasured common causes for missingness indicators}
  \label{m-DAG.D}
\end{figure*}

\begin{figure*}[!htbp]
  \centering
  \includegraphics[width=0.5\linewidth]{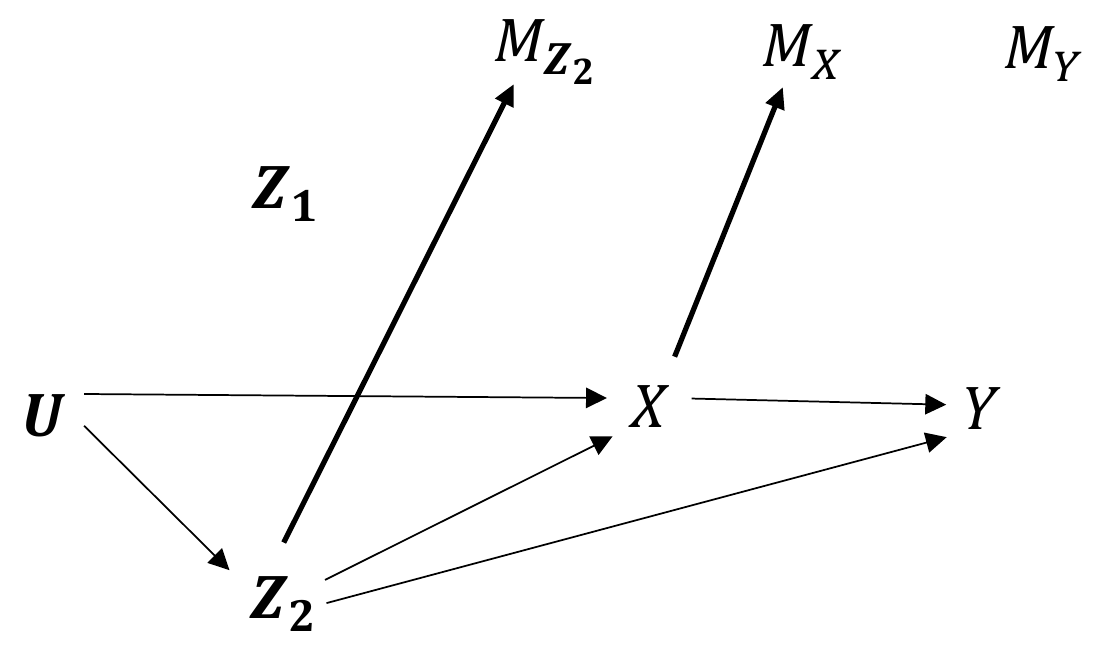}
  \caption{m-DAG D'}
  \label{m-DAG.D'}
\end{figure*}

\begin{figure*}[!htbp]
  \centering
  \includegraphics[width=0.5\linewidth]{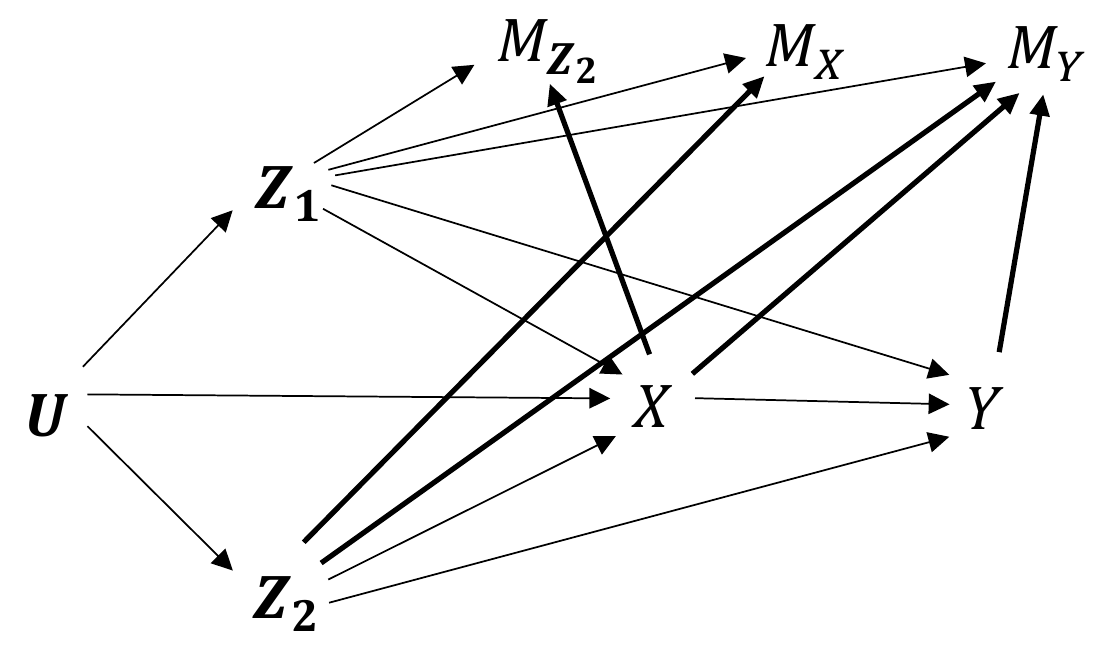}
  \caption{Canonical m-DAG G for a general point-exposure study assuming no unmeasured common causes for missingness indicators}
  \label{m-DAG.G}
\end{figure*}

\begin{figure*}[!htbp]
  \centering
  \includegraphics[width=0.5\linewidth]{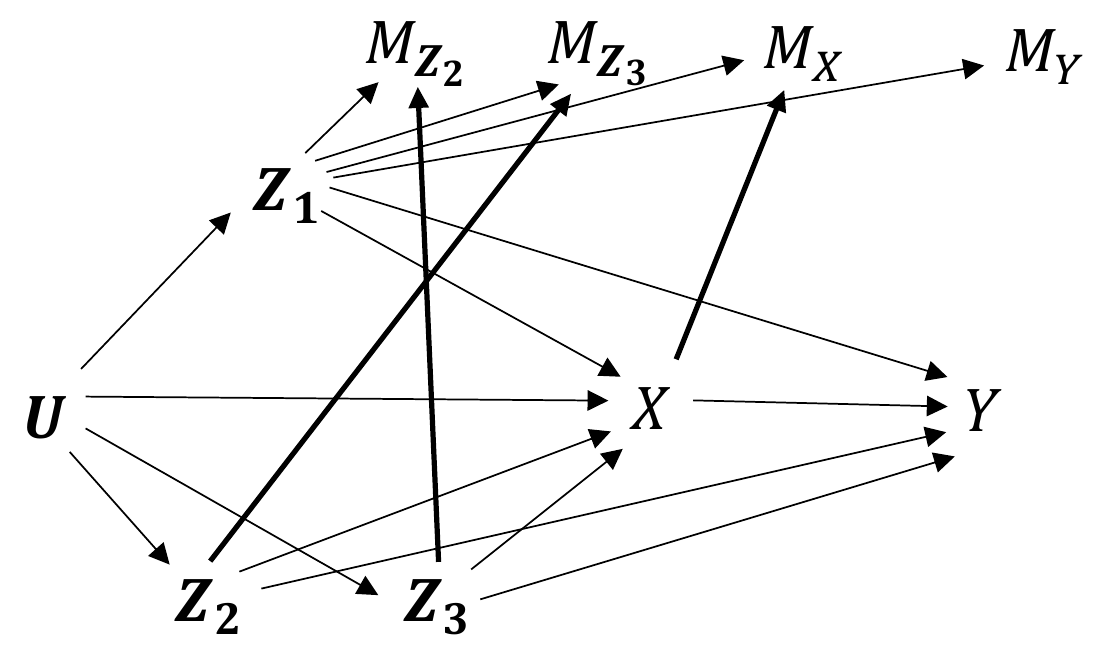}
  \caption{Canonical m-DAG D'' for a general point-exposure study assuming no unmeasured common causes for missingness indicators}
  \label{m-DAG.D''}
\end{figure*}

\bibliographystyle{unsrtnat}
\bibliography{Main}  






\end{document}